\documentclass[unknownkeysallowed]{article}
\usepackage[left=2.5cm,top=2.5cm,right=2.5cm,bottom=2.5cm]{geometry}

\usepackage{amsmath, amssymb, amsthm}
\usepackage{mathrsfs}
\usepackage{algorithm, algpseudocode}
\usepackage{tabularx, environ}
\usepackage{xspace}
\usepackage{enumitem}
\usepackage[UKenglish]{babel}
\usepackage{tikzscale,tikz,standalone,soul} 
\usetikzlibrary{shapes.geometric,calc}

\definecolor{vaxgreen}{RGB}{59,125,35}
\definecolor{contagred}{RGB}{192,0,0}

\newcommand{\circnode}[2]{%
  \begin{scope}[shift={(#1)}]
    \fill[vaxgreen!55!white, draw=black, line width=0.8pt] circle (0.42);
    \node[black] at (0,0) {#2};
  \end{scope}
}

\newcommand{\circnodeinoc}[2]{%
  \begin{scope}[shift={(#1)}]
    \fill[vaxgreen!25!white, draw=vaxgreen, line width=1.3pt] circle (0.44);
    \node[black, font=\bfseries] at (0,0) {#2};
  \end{scope}
}

\newcommand{\chevron}[2]{%
  \begin{scope}[shift={#1}, rotate=#2]
    \fill[white, draw=vaxgreen, line width=1pt]
      (-0.22,0.28) -- (0.22,0) -- (-0.22,-0.28) -- (-0.08,0) -- cycle;
  \end{scope}
}

\definecolor{cYellow}{RGB}{247,230,120}
\definecolor{cRed}{RGB}{224,110,100}
\definecolor{cBlue}{RGB}{120,175,215}
\definecolor{cGreen}{RGB}{110,165,120}

\colorlet{fillYellow}{cYellow!55!white}
\colorlet{fillRed}{cRed!60!white}
\colorlet{fillBlue}{cBlue!55!white}
\colorlet{fillGreen}{cGreen!55!white}

\colorlet{cYellowText}{cYellow!70!black}

\newcommand{\gridcircle}[3]{%
  \begin{scope}[shift={(#1)}]
    \fill[#2, draw=black, line width=0.8pt] circle (0.42);
    \node[black] at (0,0) {#3};
  \end{scope}
}

\newcommand{\gridsplitcircle}[4]{%
  \begin{scope}[shift={(#1)}]
    \begin{scope}
      \clip (-0.42,0) rectangle (0.42,0.42);
      \fill[#2] (0,0) circle (0.42);
    \end{scope}
    \begin{scope}
      \clip (-0.42,-0.42) rectangle (0.42,0);
      \fill[#3] (0,0) circle (0.42);
    \end{scope}
    \draw[black, line width=0.8pt] (0,0) circle (0.42);
    \node[black] at (0,0) {#4};
  \end{scope}
}

\newcommand{\gridsquare}[2]{%
  \begin{scope}[shift={(#1)}]
    \fill[white, draw=black, line width=0.8pt] (-0.36,-0.36) rectangle (0.36,0.36);
    \node[black] at (0,0) {#2};
  \end{scope}
}

\newcommand{\gridroot}[1]{%
  \begin{scope}[shift={(#1)}]
    \fill[white, draw=black, line width=0.8pt] (-0.36,-0.36) rectangle (0.36,0.36);
    \draw[black, line width=1pt] (-0.22,-0.22) -- (0.22,0.22);
    \draw[black, line width=1pt] (-0.22,0.22) -- (0.22,-0.22);
  \end{scope}
}

\newcommand{\gridspecial}[4]{%
  \begin{scope}[shift={(#1)}]
    \fill[#2!25!white, draw=#2, line width=1.3pt] circle (0.44);
    \node[#3, font=\bfseries] at (0,0) {$#4$};
  \end{scope}
}

\newcommand{\gridchevron}[3]{%
  \begin{scope}[shift={#1}, rotate=#2]
    \fill[white, draw=#3, line width=1pt]
      (-0.22,0.28) -- (0.22,0) -- (-0.22,-0.28) -- (-0.08,0) -- cycle;
  \end{scope}
}

\newcommand{\gridchevronsplit}[4]{%
  \begin{scope}[shift={#1}]
    \begin{scope}[rotate=#2]
      \fill[white] (-0.22,0.28) -- (0.22,0) -- (-0.22,-0.28) -- (-0.08,0) -- cycle;
    \end{scope}
    \begin{scope}
      \clip (-0.5,0) rectangle (0.5,0.5);
      \begin{scope}[rotate=#2]
        \draw[#3, line width=1pt] (-0.22,0.28) -- (0.22,0) -- (-0.22,-0.28) -- (-0.08,0) -- cycle;
      \end{scope}
    \end{scope}
    \begin{scope}
      \clip (-0.5,-0.5) rectangle (0.5,0);
      \begin{scope}[rotate=#2]
        \draw[#4, line width=1pt] (-0.22,0.28) -- (0.22,0) -- (-0.22,-0.28) -- (-0.08,0) -- cycle;
      \end{scope}
    \end{scope}
  \end{scope}
}

\newcommand{\cellVone}[2]{\gridspecial{#1}{cRed}{black}{V_1}}
\newcommand{\cellVtwo}[2]{\gridspecial{#1}{cBlue}{black}{V_2}}
\newcommand{\cellVthree}[2]{\gridspecial{#1}{cGreen}{black}{V_3}}
\newcommand{\cellVfour}[2]{\gridspecial{#1}{cYellow}{black}{V_4}} 

\newtheorem{theorem}{Theorem}[section]
\newtheorem{lemma}[theorem]{Lemma}
\newtheorem{corollary}[theorem]{Corollary}

\newtheorem{ques}[theorem]{Question}

\theoremstyle{definition}
\newtheorem{definition}[theorem]{Definition}
\theoremstyle{definition}

\makeatletter
\newcommand{\name}[1]{\gdef\@name{#1}}
\newcommand{\instance}[1]{\gdef\@instance{#1}}
\newcommand{\question}[1]{\gdef\@question{#1}}
\NewEnviron{problem}{
  \name{}\instance{}\question{}
  \BODY
  \par\addvspace{.5\baselineskip}
  \noindent
  \fbox{
    \begin{minipage}{0.9\textwidth}
    \begin{tabularx}{\textwidth}{@{\hspace{\parindent}} l X c}
        \noalign{\vskip 1mm}  
        \multicolumn{2}{@{\hspace{\parindent}}l}{\Large\scshape\@name} \\
        \noalign{\vskip 2mm}  
        \textbf{Instance:} & \@instance \\
        \textbf{Question:} & \@question
    \end{tabularx}
    \end{minipage}
  }
  \vskip 1mm
}
\makeatother

\floatname{algorithm}{Procedure}

\newcommand*{\Fire}{\textup{\textsc{Firefighter}}\xspace} 
\newcommand*{\bFire}{$b$-\Fire} 

\newcommand*{\IV}{\textup{\textsc{InfectVax}}\xspace}

\makeatletter
\newcommand*{\IVLong}{%
  \ifnum\spacefactor>2000
    The Infectious Vaccination problem%
  \else
    the Infectious Vaccination problem%
  \fi
  \xspace
}
\makeatother

\newcommand*{\nat}{\in\mathbb{N}}
\newcommand*{\bigO}[1]{\mathcal{O}({#1})}

\DeclareMathOperator{\strength}{str}
\DeclareMathOperator{\inoc}{inoc}

\title{The Infectious Vaccination Problem: a variant of Firefighting with Spreading Defence}

\author{
Jessica Enright
   \thanks{School of Computing Science, University of Glasgow;
   \texttt{Jessica.Enright@glasgow.ac.uk}
   }
\and Melissa A. Huggan
   \thanks{Department of Mathematics and Computer Science, University of Lethbridge;
   \texttt{melissa.huggan@uleth.ca} 
   }
\and Ethan Hunter-Frankland
   \thanks{School of Mathematics, University of Birmingham;
  \texttt{e.hunter-frankland@bham.ac.uk}
   }
\and M.E. Messinger
   \thanks{Department of Mathematics and Computer Science, Mount Allison University;
   \texttt{mmessinger@mta.ca}
   }
\and Dylan Pearson
  \thanks{Department of Mathematics and Computer Science, Mount Allison University.
   }
}

\date{}

\begin{document}

\maketitle

\begin{abstract}
The Firefighter Problem models a spreading process (originally a fire, alternatively an infection or rumour, for example) on a graph. A defender saves a single vertex per turn; after each defence, the fire spreads to the unburned and undefended neighbours of all burning vertices. Deciding whether a strategy exists for the defender to protect some targeted number of vertices is computationally hard in graphs in general, but tractable in some restricted cases. 

Inspired by research into spreadable rabies vaccines for bats, we study a variant of the Firefighter problem in which defence also spreads. Some approximation results are already known for this problem; we provide algorithmic and hardness results, as well as containment results for the infinite $n$-dimensional Cartesian and strong grid graphs.
\end{abstract}


\section{Introduction}\label{sec:background}

The Firefighter problem (\Fire) is a deterministic, discrete-time game that models the spread of fire on a graph.  Originally introduced by Harnell~\cite{hartnell_1995}, \Fire and its variants have been well-studied; see the survey by Finbow and MacGillivray~\cite{finbow_2009}.

In a game of \Fire, vertices can be in one of three states: \emph{open, burning} or \emph{defended.} All vertices are initially open. However, once burning or defended, a vertex remains so for the rest of the game. 
An instance of \Fire on a rooted graph $G$ proceeds as follows. At time-step $t=0$, the fire breaks out at root $r$. For each time-step $t>0$, the firefighter selects a vertex to defend, and then the fire spreads to the open neighbours of all burning vertices. If there is some time-step at which no open vertex is adjacent to a burning vertex, we declare the fire to be \emph{contained} and count the number of \emph{saved} vertices, those either defended or open and therefore unreachable by the fire. The number of time-steps from outbreak to containment of a particular instance is the \emph{lifetime} of the instance.  A \emph{certificate} for an instance of \Fire is a \emph{strategy}, a sequence of vertices selected for defence (in order of when they are defended).

In the generalised model, \bFire, up to $b$ vertices are defended at each time-step for some fixed budget $b$. We consider a variant of \Fire and \bFire called \IVLong (\IV). As \IV is inspired by studies of spreadable --or self-disseminating-- animal vaccines in combatting viruses, we refer to the fire as the \emph{contagion} that spreads during each time-step. All defended vertices are called \emph{vaccinated} and those that are defended directly are called \emph{inoculated.} \IV proceeds as \bFire with an inoculation budget $b$ per time-step, but the vaccination spreads in the same way as the contagion. That is, for \IV: all vertices are initially open and at time-step $0$, the contagion breaks out at a root $r$ of graph $G$.  During each time-step $t > 0$,
\begin{enumerate}
    \item the vaccine spreads from all vaccinated vertices to their open neighbours;
    \item up to $b$ vertices are inoculated, for some integer budget $b$; and
    \item the contagion spreads to the open neighbours of all contagious vertices. 
\end{enumerate}
An instance of \IV is composed of a graph $G$, a root $r\in V(G)$, a budget $b\in\mathbb{N}$ and a target number of vertices to save $k\in\mathbb{N}$. We represent an instance of \IV as the tuple $(G, r, b, k)$.
Interestingly, \IV has been studied as the \emph{spreading vaccination model} by Anshelevich et al.\ \cite{anshelevich_2012} and Floderus et al.\ \cite{floderus_2011}. Both studied it as an optimisation problem with two objectives: maximising the number of saved vertices and minimising the budget used.
Anshelevich et al.\ \cite{anshelevich_2012} prove the existence of constant-factor approximation algorithms for maximising saved vertices, and use work by Floderus et al.\ \cite{floderus_2011} to show that there is no $((1 - o(1)) \ln n)$ approximation (unless $NP \subseteq DTIME(n^{\texttt{poly}\ln(n)})$) for minimising the budget, but that there is a polynomial-time $\ln n$-factor approximation algorithm for this objective \cite{anshelevich_2012}.
They observe that \IV is easier to approximate than \Fire.

We begin with terminology and formal definitions in Section~\ref{sec:problem_definition}.  In Section~\ref{sec:alg}, we show that for an interval graph $G$ on $n$ vertices and $m$ edges, \IV on input $(G,r,b,k)$ can be solved in $\mathcal{O}(n+m)$ time.  In Section~\ref{sec:parameterize}, we consider a parameterised approach and show that \IV is FPT (fixed-parameter tractable) when parameterised by budget and vertex cover size, or by budget and neighbourhood diversity. We further show in Section~\ref{sec:bounded-path-length}, that \IV is in XP when parameterised by maximum path length.  With budget $b=1$, we study two vertex-transitive families of infinite graphs in Section~\ref{sec:contain}: the infinite $n$-dimensional Cartesian grid and the infinite $n$-dimensional strong grid.  We provide strategies which guarantee containment on both graph families.  Since the number of vertices saved is infinite, we bound the number of time-steps until containment occurs, which in turn bounds the number of contagious vertices.  Finally, in Section~\ref{sec:subgraphs}, we show that for a fixed budget $b$ the number of inoculations required to contain the virus is not closed under induced subgraphs.

\section{Preliminaries}\label{sec:problem_definition}

We begin with essential graph theoretic definitions and recommend \cite{wilson_1996} for a more complete overview.  We assume all graphs to be undirected and unweighted. An \emph{induced subgraph} $H=(V_H, E_H)$ of $G=(V_G, E_G)$, denoted $G[V_H]$, is a subgraph of $G$ on the vertices $V_H$ such that for all $v, w\in V_H$, if $(v, w)\in E_G$, then $(v, w)\in E_H$. If the graph $G$ is clear from the context, we use simplified notation $G=(V,E)$.
For $v\in V(G)$, let $N(v)$ and $N[v]$ denote the open and closed neighbourhoods of $v$, respectively. 

A \emph{strategy} for a game of \Fire is a sequence of vertices in which the $i^{th}$ vertex is selected for defence at the $i^{th}$ time-step. 
In \IV, we take as input an integer \emph{budget} $b$ that bounds the number of inoculations (direct defences) per time-step. Thus, strategies are sequences of sets of vertices. The following defines such strategies for variants such as \IV and \bFire.

\begin{definition}[Strategy with multiple defences per time-step]\label{def:strategy}
    For an instance of a \Fire variant on a rooted graph $G$ in which some $b\nat$ vertices can be defended per time-step, a strategy is a sequence of $\ell$ sets of vertices 
    $$\sigma = (\sigma_1,\sigma_2,\dots,\sigma_{\ell})$$
 such that $\sigma_i \subseteq V(G)$ is the set of vertices defended during time-step $i$ and $|\sigma_i| \leq b$.
 
    The \emph{length} of the strategy is the number of sets in the sequence.
\end{definition}

Using this definition, \bFire is formally defined as follows \cite{bazgan_2013}. 

\begin{problem}
    \name{\bFire}
      \instance{A rooted graph $(G,r)$, an integer budget $b\geq 1$ and an integer $k\geq 1$.}
      \question{Is there a strategy $\sigma$ such that if the fire breaks out at $r$:
    \begin{itemize}[nosep]
        \item for each $\sigma_i\in \sigma$, by time $i$ each vertex in $\sigma_i$ is neither burning nor defended,
        \item the number of vertices defended in each time-step is at most $b$,
        \item after time $|\sigma|$, no open vertex is adjacent to a burning vertex, and
        \item at time $|\sigma|$, at least $k$ vertices are saved by implementing $\sigma$?
    \end{itemize}}
\end{problem}

For \bFire on infinite graphs, a common consideration is whether the fire can be \emph{contained,} which is to ask whether we can deploy defences to enclose the fire to a finite set of vertices within a finite number of time-steps. 
When we say the fire is \emph{contained} on a finite graph, we require that, when the fire can no longer spread, there is at least one undefended and unburned vertex in the same component as the root vertex but inaccessible by the fire because of the defences placed. This intentionally excludes the case in which the fire cannot spread further because all undefended vertices of a finite graph have been burned.  The following is attributed to Moeller and Wang \cite{moeller_2002}.
\begin{theorem}\label{thm:fire-contain-cartesian}
    On the infinite $2$-dimensional Cartesian grid, containment is not possible for \Fire but is possible for \bFire with budget $b\geq 2$ in fewer than 8 time-steps.
\end{theorem} 

In stark contrast, we show in Section \ref{sec:contain} that the contagion can be contained for \IV on the infinite $n$-dimensional Cartesian grid with a budget of $b=1$ per time-step by the end of time-step $n^2+1$.

Next we formally define \IV.  Recall that we view the vertices of a graph as being in one of three states at the end of each time-step: open, contagious, or vaccinated.  Once a vertex is in the contagious or vaccinated state, it remains in that state forever.  A vertex that is neither contagious nor vaccinated is open.  A further subtlety we note is that inoculated vertices are vaccinated, but vaccinated vertices are not necessarily inoculated.

On a graph $G=(V, E)$ at time $t=0$, the contagion breaks out at vertex $r\in V$. At each subsequent time-step, all open vertices adjacent to vaccinated vertices become vaccinated, then we choose up to $b$ additional open vertices to inoculate, and then the contagion propagates to all open neighbours of contagious vertices. This continues until the contagion cannot spread. 

\begin{problem}
    \name{\IV}
      \instance{A graph $G=(V, E)$, a root vertex $r\in V$, a budget $b\in\mathbb{N}$ and target $k\in\mathbb{N}$.}
      \question{If the contagion breaks out at $r$ in time-step $0$ and in each subsequent time-step $t>0$:
      \begin{itemize}[nosep]
          \item each open neighbour of each vaccinated vertex becomes vaccinated, then
          \item we select up to $b$ open vertices to inoculate, and then
          \item each open neighbour of each contagious vertex becomes infected,
      \end{itemize}
      is there a strategy $\sigma$, 
      such that:
      \begin{itemize}[nosep]
        \item after time $|\sigma|$, no open vertex is adjacent to a contagious vertex, and
        \item at least $k$ vertices are saved by time $|\sigma|$?
      \end{itemize}}
\end{problem}

\smallskip

As an example, suppose the contagion breaks out at the origin $(0,0)$ on an $n \times n$ Cartesian grid, where $n \geq 8$.  In Figure~\ref{fig:contain-grid1}, the origin is indicated by the red circled vertex labelled ``$0$".  Fix $b=1$, so at most one vertex will be inoculated during each time-step and let $\sigma = ( \{(1,0)\}, \{(-2,0)\}, \{(0,3)\}, \{(0,-4)\}, \emptyset)$.    The first event during each time-step is  the vaccination spread:\ during time-step $1$, there is no vaccination to spread.  The second event during each time-step is that $b$ vertices are inoculated:\ during time-step $1$, the solid green vertex at $(1,0)$, labelled ``$1$'' is inoculated.  The third event during each time-step is the spread of the contagion:\ during time-step $1$, the contagion spreads to the three red circled vertices labelled ``$1$".  Continuing this process, the reader can see that after the vaccine spreads during time-step $5$, there is no vertex to which the contagion can spread and no need to inoculate a fifth vertex. Thus some of the $\sigma_i$ may be empty.

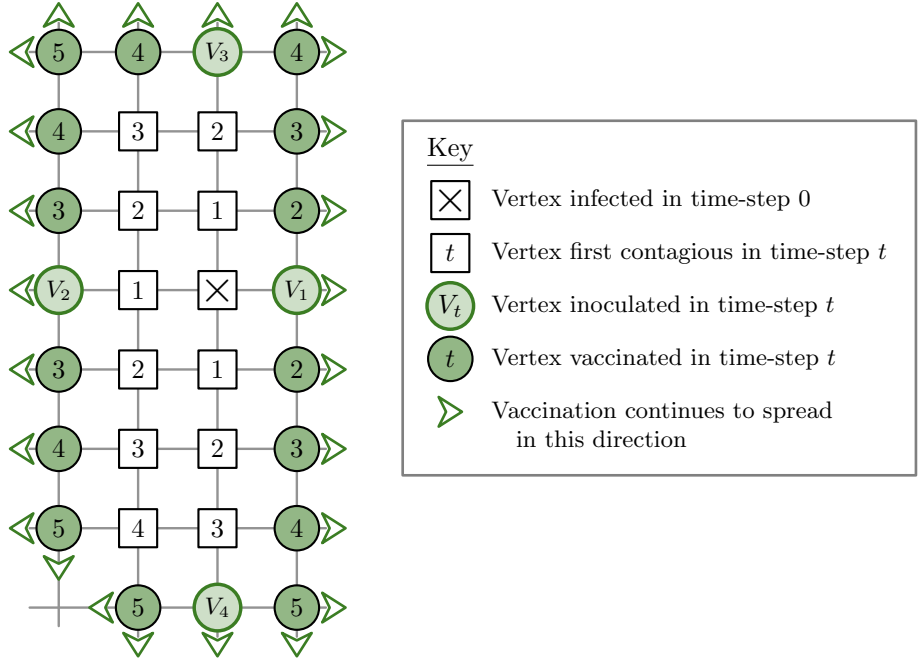
\begin{figure}[!ht]
  \begin{center}
\begin{tikzpicture}[line cap=round, line join=round,scale=0.7]

  \foreach \r in {1,...,8}{
    \foreach \c in {1,...,4}{
      \coordinate (n\r\c) at ({(\c-1)*1.5},{-(\r-1)*1.5});
    }
  }

  \foreach \r in {1,...,8}{
    \foreach \c [count=\cc from 2] in {1,2,3}{
      \draw[gray!85, line width=0.9pt] (n\r\c) -- (n\r\cc);
    }
  }
  \foreach \c in {1,...,4}{
    \foreach \r [count=\rr from 2] in {1,...,7}{
      \draw[gray!85, line width=0.9pt] (n\r\c) -- (n\rr\c);
    }
  }

  \foreach \c in {1,...,4}{
    \draw[gray!85, line width=0.9pt] (n1\c) -- ++(0,0.55);
    \chevron{($(n1\c)+(0,0.7)$)}{90}
  }
  \foreach \c in {2,3,4}{
    \draw[gray!85, line width=0.9pt] (n8\c) -- ++(0,-0.55);
    \chevron{($(n8\c)+(0,-0.7)$)}{270}
  }
  \draw[gray!85, line width=0.9pt] (n81) -- ++(0,-0.35);
  \foreach \r in {1,...,7}{   
    \draw[gray!85, line width=0.9pt] (n\r1) -- ++(-0.55,0);
    \chevron{($(n\r1)+(-0.7,0)$)}{180}
  }
  \draw[gray!85, line width=0.9pt] (n81) -- ++(-0.55,0); 
  \chevron{($(n81)+(0.8,0)$)}{180} 

  \foreach \r in {1,...,8}{
    \draw[gray!85, line width=0.9pt] (n\r4) -- ++(0.55,0);
    \chevron{($(n\r4)+(0.7,0)$)}{0}
  }

  \circnode{n11}{5}      
  \circnode{n12}{4}
  \circnodeinoc{n13}{\small{$V_3$}}  
  \circnode{n14}{4}

  \circnode{n21}{4}      
  \gridsquare{n22}{3}
  \gridsquare{n23}{2}                     \circnode{n24}{3}

  \circnode{n31}{3}      
  \gridsquare{n32}{2}
  \gridsquare{n33}{1}                     
  \circnode{n34}{2}

  \circnodeinoc{n41}{\small{$V_2$}}  
  \gridsquare{n42}{1}
  \gridsquare{n43}{\LARGE $\times$}                     \circnodeinoc{n44}{\small{$V_1$}}

  \circnode{n51}{3}      
  \gridsquare{n52}{2}
  \gridsquare{n53}{1}                     \circnode{n54}{2}

  \circnode{n61}{4}      
  \gridsquare{n62}{3}
  \gridsquare{n63}{2}                     
  \circnode{n64}{3}

  \circnode{n71}{5}      
  \gridsquare{n72}{4}
  \gridsquare{n73}{3}                     \circnode{n74}{4}

  \circnode{n82}{5}
  \circnodeinoc{n83}{\small{$V_4$}}
  \circnode{n84}{5}

  \chevron{($(n71)!0.5!(n81)$)}{270}

  \begin{scope}[xshift=6.5cm, yshift=-8cm]
    \def\w{9.7} \def\h{6.7} 
    \draw[draw=gray, line width=1pt] (0,0) rectangle (\w,\h);
    \node at (0.9,\h-0.6) {\underline{Key}};

    \coordinate (k0) at (0.9,\h-1.5); 
    \coordinate (k1) at (0.9,\h-2.5);
    \coordinate (k2) at (0.9,\h-3.5);
    \coordinate (k3) at (0.9,\h-4.5);
    \coordinate (k4) at (0.9,\h-5.5);

    \gridsquare{k0}{\LARGE $\times$}
    \node[right, align=left] at (1.5,\h-1.5) {\small Vertex infected in time-step $0$ };   
    \gridsquare{k1}{$t$}
    \node[right, align=left] at (1.5,\h-2.5) {\small Vertex first contagious in time-step $t$};
    \circnodeinoc{k2}{$V_t$}
    \node[right, align=left] at (1.5,\h-3.5) {\small Vertex inoculated in time-step $t$};
    \circnode{k3}{$t$}
    \node[right, align=left] at (1.5,\h-4.5) {\small Vertex vaccinated in time-step $t$};
    \chevron{(k4)}{0}
    \node[right, align=left] at (1.5,\h-5.5) {\small Vaccination continues to spread};
    \node[right, align=left] at (1.5,\h-6) {\small\quad in this direction};
    
  \end{scope}

  \end{tikzpicture}
\end{center}
\caption{%
Using budget $b=1$, a strategy for \IV on the infinite $2$-dimensional Cartesian grid.}
\label{fig:contain-grid1}
\end{figure}

Like \Fire and \bFire, we also consider whether containment can occur for \IV on infinite graphs.  Thus, in Section 5 where containment is the objective, we omit $k$ from consideration (since an infinite number of vertices will be saved by containment) and regard an instance of \IV to be the triple $(G,r,b)$. In the case where containment occurs for a particular value of $b$, we will seek the minimum number of inoculations needed for containment occur.

\section{Algorithmic and hardness results}\label{sec:alg}

In this section, we complement an existing hardness result on \IV with several tractability results on restrictive graph classes inspired by parallel results on more classical Firefighting problems.  

Anshelevich et al.\ \cite{anshelevich_2012} observed that \IV is equivalent to \bFire when restricted to trees.
Informally, this is because a vertex in a tree is saved if it has at least one vaccinated (defended)  ancestor, so the spread of the vaccine does not change the number of saved (vaccinated) vertices.  Thus, existing results for \bFire on trees follow immediately for \IV on trees, including NP-completeness.
\begin{theorem}[\cite{anshelevich_2012}]
    \IV is NP-complete.
\end{theorem}

A tractability result and an approximation result immediately following from Bazgan et al.\  \cite{bazgan_2013} and Cai et al.\ \cite{cai_2008} respectively:

\begin{theorem}
    \IV with budget $b$ is solvable in time polynomial in input tree size on trees of maximum degree $b+2$, rooted at a vertex of degree $b+1$.
\end{theorem}

\begin{theorem}
    There is a $(1-1/e)$-approximation algorithm for \IV on trees.
\end{theorem}

We next consider interval graphs. Let $G=(V, E)$ be a graph on $n$ vertices such that there exists a set $I=\{I_1, \dots, I_n\}$ of intervals on the real line and a bijection $\mathcal{I}:V\rightarrow I$ that maps each vertex $v\in V$ to an interval $\mathcal{I}(v)=I_v$ such that, for all $u \in N(v)$, $I_u\cap I_v\neq\emptyset$ and for all $w\notin N(v)$, $I_v\cap I_w = \emptyset$. Such a pair ($I$, $\mathcal{I}$) is called an \emph{interval representation} and a graph for which there exists an interval representation is called an \emph{interval graph}. Given an interval graph $G=(V, E)$, we  can obtain an interval representation in $O(|V|+|E|)$ time \cite{booth_1976}.

Adapting an approach of Fomin et al. \cite{fomin_2016}, we show that in any instance of \IV on an interval graph, the contagion can be contained within at most three time-steps and we can resolve \IV in polynomial time on interval graphs. We essentially do this by first playing optimally on one side of the contagion, and then on the other. As in \cite{fomin_2016}, we assume without loss of generality that interval endpoints are unique, since intervals can be perturbed slightly if necessary.

\begin{theorem}\label{thm:iv-interval}
Given an instance $(G,r,b,k)$ of \IV on an interval graph $G$ with $n$ vertices and $m$ edges, this instance can be solved in $\mathcal{O}(n+m)$ time.
\end{theorem}

\begin{proof}
First, if $b \geq N(r)$ then we inoculate the neighbourhood of $r$ and immediately restrict the contagion to a single vertex.  

Second, if $N(r) \subseteq N[v]$ for some vertex $v$ then during time-step 1, we inoculate $v$ along with up to $b-1$ neighbours of $r$.  This restricts the contagion to  $\max(1,|N[r]|-b)$ vertices if $v \in N(r)$, and $\max(1,|N[r]|-b+1)$ otherwise; both are best possible.

Otherwise, we observe that there can be no single neighbour of $r$ whose interval extends both farther to the left and farther to the right than the interval of any other neighbour of $r$. Let $v_\ell$ and $v_r$ be the unique neighbours of $r$ such that the left endpoint of $\mathcal{I}(v_\ell)$ extends farther to the left than the interval of any other neighbour of $r$ and the right endpoint of $\mathcal{I}(v_r)$ extends farther to the right than the interval of any other neighbour of $r$.

If $b \geq 2$ then we can solve the instance optimally by inoculating both $v_\ell$ and $v_r$ as well as up to $b-2$ other distinct members of $N(r)$ in time-step 1, therefore restricting the contagion to $\max(1,|N[r]|-b)$ vertices.

The remaining case is that $b=1$ and $v_\ell$ and $v_r$ are distinct. We argue that we can contain the contagion in three time-steps, in a way that is a best-possible solution. For this, we need to refer to the neighbours of neighbours of $r$ that lie to the left of $r$ as separate from those on the right.  We therefore denote by $L_2(r)$ the set of vertices that are distance two from $r$ and correspond to intervals that lie entirely to the left of $I(r)$ in $\mathcal{I}$. And symmetrically, we denote by $R_2(r)$ the set of vertices that are distance two from $r$ and correspond to intervals that lie entirely to the right of $I(r)$ in $\mathcal{I}$.

For $b=1$, we next argue that if $v_\ell$ and $v_r$ are distinct then under any inoculation strategy at least $|N(r)| + \min(|L_2(r)|, |R_2(r)|) -1$ vertices will be infected, and we then argue that we can find a strategy that achieves this in the required time complexity. We consider two cases, depending on whether a vertex in $N(r)$ is inoculated in time-step 1 or not.

If not, then $|N(r)|$ vertices become infected in time-step 1.  Because the vertex inoculated in time-step 1 was not a neighbour of $r$, its interval must lie entirely to the right or entirely to the left of $\mathcal{I}(r)$:  without loss of generality, entirely to the right.  Observe that when the vaccine spreads at the beginning of time-step 2, it cannot spread to any member of $L_2(r)$.   During the inoculation-selection phase in time-step 2, at most one member of $L_2(r)$ is inoculated, and then at the end of that time-step all other members will become infected. Therefore at least $|N(r)| + |L_2(r)| -1 + |\{r\}| = |N(r)| + |L_2(r)|$ vertices are infected.

Suppose we inoculate within $N(r)$ during time-step 1. Then at most one of $v_r$, $v_\ell$ is inoculated in time-step 1 and $|N(r)|-1$ vertices become infected in time-step 1.  When the vaccine spreads during time-step 2, at most one of $L_2(r)$ or $R_2(r)$ become vaccinated: without loss of generality, suppose the vaccination does not spread to $L_2(r)$ during time-step 2. Then no neighbour of a vertex of $L_2(r)$ was inoculated on time-step 1. During the inoculation-selection phase in time-step 2, at most one vertex of $L_2(r)$ can be inoculated and then at the end of that time-step all other vertices of $L_2(r)$ will become infected.  At least $|N(r)| + |L_2(r)| -1$ vertices are infected by the end of time-step 2. 

Thus, at least $|N(r)| + \min(|L_2(r)|, |R_2(r)|) -1$ will be infected.  This bound can be achieved by the following inoculation strategy:
\begin{itemize}
\item in time-step 1: inoculate $v_\ell$ if $|L_2(r)| > |R_2(r)|$, else inoculate $v_r$;
\item in time-step 2: inoculate the vertex of $L_2(r)$ with the leftmost left endpoint if we previously inoculated $v_r$, else inoculate the vertex of $R_2(r)$ with the rightmost right endpoint.  
\end{itemize}
With the above strategy, the contagion is contained after the vaccine spreads during time-step 3, with $|N(r)| + \min(|L_2(r)|, |R_2(r)|) -1$ vertices infected.  Determining which of $L_2(r)$ or $R_2(r)$ is larger, and identifying the vertices with extremal intervals that we require can be achieved from an adjacency list and the interval representation in $\bigO{|V|+|E|}$ time. 
\end{proof}

\section{Parameterised approaches}\label{sec:parameterize}

If a problem $A$ parameterised by $k$ can be solved in $\bigO{n^{f(k)}}$ time for input size $n$ and a computable function $f$, we say $A$ is in XP, meaning \emph{slicewise polynomial} \cite{downey_1999}. More strongly, $A$ is \emph{fixed-parameter tractable} (FPT) with respect to $k$ if it can be solved in $\bigO{f(k)n^{\bigO{1}}}$ time.  
We show that \IV is FPT when parameterised by the combination of budget with vertex cover size (Section \ref{sec:vertex-cover}) or neighbourhood diversity (Section \ref{sec:neighbourhood-diversity}).  It is in XP when parameterised by the combination of budget and maximum path length (Section \ref{sec:bounded-path-length}).

\begin{lemma}\label{lem:fpt-simulation}
  Let $\mathcal{P}$ be a parameterised decision problem on a graph $G=(V, E)$ on $n$ vertices for which a certificate is a finite sequence of subsets of vertices $\sigma=(\sigma_t)_{t\leq T}$ for some $T\in\mathbb{N}$. Let $f, g$ and $h$ be computable functions and suppose:
  \begin{enumerate}
    \item The lifetime of any instance is bounded by $f(k)$ (that is, $T\leq f(k)$). \label{it:1}
    \item For each time-step $t$, $V$ can be partitioned into at most $g(k)$ equivalence classes such that selecting any two vertices from the same class for $\sigma_t$ yields identical outcomes for the rest of the process. \label{it:2}
    \item The outcome of a fixed certificate $\sigma$ can be simulated in $h(n,k)$ time.\label{it:3}
  \end{enumerate}
  Then $\mathcal{P}$ can be solved in $\big(g(k)\big)^{f(k)} \cdot h(n,k)$ time. In particular, if $\mathcal{P}$ can be expressed as $F(k)\cdot \mathrm{poly}(n)$ for some computable $F$, then $\mathcal{P}$ is fixed-parameter tractable.
\end{lemma}

\begin{proof}
  By \ref{it:1}., any strategy has length at most $f(k)$. By \ref{it:2}., there are at most $g(k)$ outcome-distinct choices per step, hence at most $g(k)^{f(k)}$ candidate certificates overall. By \ref{it:3}., simulating a fixed candidate takes $h(n,k)$ time. Hence the total runtime is $g(k)^{f(k)}\cdot h(n,k)$.
\end{proof}

Two vertices $u,v$ in graph $G$ have the same {\it neighbourhood type} if $N(u) \backslash \{v\} = N(v) \backslash \{u\}$.  A {\it type partition} of $G$ is a partition $\mathcal{T} = \{T_1,\dots,T_w\}$ of $V(G)$ such that if $u$ and $v$ have the same neighbourhood type, then  $u, v \in T_i$, for some $i: 1 \leq i \leq w$.  This is an equivalence relation, whose classes are the types.  Suppose that two strategies $\sigma$ and $\sigma'$ are identical except during some time-step $t$, $u \in T_i$ is inoculated via $\sigma$, whereas $v \in T_i$ via $\sigma'$.  If the entire evolution is identical up to a permutation of vertices within $T_i$, then we say the strategies produce {\it equivalent outcomes}.   That is, if the final partition into contagious, vaccinated, and open states is the same up to relabelling within $T_i$, then $\sigma$ and $\sigma'$ produce equivalent outcomes.

\begin{lemma}\label{lem:type-equivalence}
  Let $G$ be a graph that admits a type partition $\mathcal{T}=\{T_1,\dots,T_w\}$. In an instance of \IV on $G$, inoculating different vertices from the same type yields equivalent outcomes.
\end{lemma}

\begin{proof}  Let $G=(V,E)$ be a graph with a type partition $\mathcal{T} = \{T_1, \dots, T_w\}$. Let $\sigma$ and $\sigma^\prime$ be two strategies that are identical for the first $t$ time-steps, with the exception that during time-step $t$, open vertex $u \in T_i$ is inoculated by $\sigma$ and open vertex $v \in T_i$ is inoculated by $\sigma^\prime$. After $t-1$ time-steps, the state configurations are identical; that is, regardless of whether one follows strategy $\sigma$ or $\sigma^\prime$, at the end of time-step $t-1$, vertex states are the same. By definition, $G[T_i]$ is either a clique or an independent set. Consequently, during time-step $t$, the spread of the vaccine and the contagion to vertices outside $T_i$  occurs identically, regardless of whether $u$ or $v$ is inoculated at time-step $t$. In $T_i$, all vertices other than $u$ and $v$ have identical neighbourhoods and identical states at the start of the time-step, so the spread of the vaccine or contagion within $T_i$ is also identical under both strategies. Hence, at the end of time-step $t$, the configurations remain identical up to permutation of vertices within $T_i$.  Consequently, inoculating different vertices from the same type produces equivalent outcomes.
\end{proof}

\subsection{Vertex cover}
\label{sec:vertex-cover}

A \emph{vertex cover} of a graph $G=(V,E)$ is a set $S\subseteq V$ such that every edge of $G$ has at least one endpoint in $S$. The minimum size of such a set is the \emph{vertex cover number} $s$ of $G$.  For an instance $(G, r, b, k)$ of \IV where $G=(V,E)$ has a vertex cover $S$ of size $s$, let $\mathcal{C} = \{C_1, C_2, \ldots, C_{2^s}\}$ be the collection of all possible subsets of $S$ that could appear as neighbourhoods of vertices in $V\setminus S$. 
Define $\chi_{G,S}: V \setminus S \to \mathcal{C}$ by $\chi_{G,S}(v) = N(v) \cap S$ for $v \in V \setminus S$. The equivalence classes on $V\setminus S$ induced by $\chi_{G, S}$ are called \emph{$\chi_{G,S}$-classes.} There are at most $2^s$ such classes among vertices in $V\setminus S$.

\begin{lemma}\label{lem:lifetime}
  The lifetime of an instance of \IV on a graph with vertex cover number $s$ is at most $2s+2$.
\end{lemma}

\begin{proof}
Let $S$ be a vertex cover of size $s$. Note that $V\backslash S$ is an independent set. At worst, $S$ is also an independent set, and the graph is thus bipartite. Regardless of the part in which the contagion originates, since the graph is bipartite, infection must spread back and forth between $S$ and $V\backslash S$. There is a maximum of $s$ time steps to infect all of $S$ and, at worst, it takes $s$ time steps to infect all of $V \backslash S$. Thus, accounting for the initial outbreak at time-step $0$ and a final step to certify containment yields the bound $2s+2$ on the lifetime.
\end{proof}

\begin{theorem}
  Given an instance $(G, r, b, k)$ of \IV on a graph $G$ with $n$ vertices, $m$ edges, and vertex cover of size $s$, this instance can be solved in $2^{\bigO{s^2 b}}\cdot (n+m)$ time. Thus, \IV is FPT by vertex cover number and budget.
\end{theorem}

\begin{proof}
  Observe that the $\chi_{G,S}$-classes, taken together with $S$, forms a type partition of $G$, with size bounded by $2^s+ s$.  By Lemma \ref{lem:type-equivalence}, we need only consider strategies that specify, for each time-step, which types to inoculate rather than specific vertices. At each time-step, we inoculate at most $b$ vertices. For each inoculated vertex, we choose amongst the at most $2^s + s$ types.  The number of outcome-distinct ways to choose at most $b$ vertices per time-step is at most $2^{\mathcal{O}(sb)}$: for $b>1$ we have this from $\sum_{i=0}^b \binom{s+2^s}{i} \le (s + 2^s)^b = 2^{\mathcal{O}(sb)}$, for $b=1$ we have this from $\sum_{i=0}^b \binom{s+2^s}{i} = 1+s+2^s = 2^{\mathcal{O}(sb)}$.  

   By Lemma \ref{lem:lifetime}, we have that $f(s) = 2s+2$. There are $g(s) =2^{\mathcal{O}(sb)}$ strategy choices per time-step and a simulation time of $h(n,s) = O(n+m)$. By Lemma \ref{lem:fpt-simulation}, we obtain an algorithm running in $( 2^{\mathcal{O}(sb)} )^{\,2s+2} \cdot O(n+m) = 2^{O(s^2b)} \cdot O(n+m)$ time, so \IV is FPT when parameterised by vertex cover number and budget.
\end{proof}

\subsection{Neighbourhood diversity}
\label{sec:neighbourhood-diversity}

Lampis \cite{lampis_2012} introduced a generalisation of vertex cover called {\it neighbourhood diversity}.  Recall that two distinct vertices $u,v$ are of the same {\it neighbourhood type} if $N(u) \backslash \{v\} = N(v) \backslash \{u\}$.  The {\it neighbourhood diversity} of a graph $G$ is the minimum number $w$ of sets in a partition $T_1,T_2,\dots, T_w$ of $V(G)$ such that every vertex in $T_i$ has the same neighbourhood type for $i \in \{1,2,\dots,w\}$.  A {\it type partition} or a {\it neighbourhood partition} of $G$ is a partition $\mathcal{T} = \{T_1,T_2,\dots,T_w\}$.  By definition, for any type $T_i$, the induced subgraph $G[T_i]$ is either a clique or an independent set, and for any two distinct types $T_i,T_j$, either the vertices in $T_i$ are adjacent to every vertex in $T_j$ or none.     

\begin{lemma}\label{lem:neigh-div-lifetime-bound}
    The lifetime of an instance of \IV on a graph with neighbourhood diversity $w$ is at most $w$.
\end{lemma}

\begin{proof}
  Let $\mathcal{T} = \{T_1, T_2, \ldots, T_w\}$ be a type partition of $G$. Define the \emph{type graph} $H$ as the graph with vertex set $\mathcal{T}$, where two vertices $T_i$ and $T_j$ are adjacent in $H$ if and only if in $G$ every vertex in $T_i$ is adjacent to every vertex in $T_j$. 

  Initially, only the type containing the root has an infected vertex. At each subsequent time-step, if the process has not terminated, there exists an open vertex adjacent to an infected vertex; thus there is a type $T_j$ that is adjacent in $H$ to a type $T_i$ containing an infected vertex. By definition of $H$, every open vertex of $T_j$ has an infected neighbour in $T_i$, so when the contagion spreads, all open vertices of $T_j$ become infected. If a type $T$ contains an infected vertex and $G[T]$ is a clique, all open vertices of $T$ become infected in the same substep. If instead $G[T]$ is an independent set, there is no internal spread. Hence, in each time-step the number of types that contain at least one infected vertex increases by at least one. Since there are only $w$ types in total, the process can last at most $w$ time-steps.
\end{proof}

\begin{theorem}\label{thm:iv-nd-fpt}
  Given an instance $(G, r, b, k)$ of \IV on a graph $G$ with $n$ vertices, $m$ edges, and neighbourhood diversity $w$, this instance can be solved in $2^{O(w \log(b+w))} \cdot O(n+m)$ time. Thus, \IV is FPT when parameterised by neighbourhood diversity and budget.
\end{theorem}

\begin{proof}
  By Lemma \ref{lem:type-equivalence}, we need only consider strategies that specify which neighbourhood types to target for inoculation, not specific vertices. At each time-step, we can inoculate up to $b$ vertices drawn from the $w$ types. The number of ways to choose a multiset of size at most $b$ from $w$ types is $\sum_{i=0}^b \binom{w+i-1}{i} = \binom{b+w}{b}$. By Lemma \ref{lem:neigh-div-lifetime-bound}, there are at most $w$ time-steps, so the total number of type-level strategies is at most $\binom{b+w}{b}^w$. Using the standard inequality $\binom{n}{k} \leq (en/k)^k$ (see, for example, \cite{graham_1994}) and symmetry $\binom{b+w}{b}=\binom{b+w}{w}$, we get
  \[\binom{b+w}{b}^w \le \left(\frac{e(b+w)}{\min\{b,w\}}\right)^{w\,\min\{b,w\}} = 2^{O\big(w\,\min\{b,w\}\, \log(b+w)\big)}.\]

  Applying Lemma \ref{lem:fpt-simulation} with $f(w)=w$, $g(w)=\binom{b+w}{b}$ and $h(n,w)=O(n+m)$, we obtain a runtime of $\binom{b+w}{b}^w \cdot O(n+m) \in 2^{O(w \log(b+w))}\cdot O(n+m)$.
\end{proof}

\subsection{Graphs without long induced paths}
\label{sec:bounded-path-length}

A graph is \emph{$P_\ell$-free} if it contains no induced path on $\ell$ vertices.  Note that this class is usually called $P_k$-free, but $k$ is already a variable in \Fire.  Fomin et al.\ \cite{fomin_2016} show that \Fire on $P_\ell$-free graphs on $n$ vertices is solvable in $\mathcal{O}(n^\ell)$ time. Hunter and Enright \cite{hunter_2024} follow a similar approach for \Fire with costs to inoculate each vertex. Building on these two results, we show that \IV on $P_\ell$-free graphs is also tractable for fixed budgets $b$ and maximum induced path length $\ell-1$.

\begin{theorem} 
Given an instance $(G,r,b,k)$ of \IV on a $P_{\ell}$-free graph $G$ with $n$ vertices and $m$ edges, this instance can be solved in $\mathcal{O}((n+m)\,n^{b(\ell-1)})$ time.
  \label{thm:iv-pk-free}
\end{theorem}

\begin{proof}
Since $G$ is $P_\ell$-free, no open vertex can be infected after time-step $\ell-1$. Hence we need to only enumerate inoculation strategies of length up to $b(\ell-1)$, of which there are at most $\mathcal{O}(n^{b(\ell-1)})$. For each candidate, we simulate the process and validate the outcome by a search from $r$ in $\mathcal{O}(n+m)$ time to count saved vertices at containment. The runtime of this procedure is $\mathcal{O}((n+m)\,n^{b(\ell-1)})$.
\end{proof}

\begin{corollary}\label{cor:bounded-poly}
  Solving \IV on graphs of maximum path length $\ell$ is in XP parameterised by $\ell$ and the budget.
\end{corollary}

Since \emph{cographs} and \emph{split graphs} are $P_4$-free and $P_5$-free respectively, Theorem \ref{thm:iv-pk-free} implies the existence of algorithms for solving \IV with budget $b$ on cographs and split graphs on $n$ vertices, $m$ edges, that run in $\mathcal{O}((n+m)\,n^{3b})$ and $\mathcal{O}((n+m)\,n^{4b})$ time respectively.

\section{Containment}
\label{sec:contain}

A central problem in graph searching is how to achieve a given objective, such as capturing an intruder or containing a fire, while minimising the resources used at each time-step (see the book~\cite{Bonato_book}, for example).   In this section, we fix the budget to be $b=1$ and study containment on two vertex-transitive families of infinite graphs: the $n$-dimensional Cartesian grid and the $n$-dimensional strong grid.  Since the graphs under consideration are infinite, we do not aim to bound the number of saved vertices, but instead aim to bound the number of time-steps until containment, which bounds the number of contagious vertices.  Since we do not have a target $k$ of vertices to save, we let $(G,r,b)$ be the triple consisting of a graph $G$ with root $r$ and positive integer $b$ and consider \IV on input $(G,r,b)$.

\begin{definition}\label{defn:vaccination_strategy} 
A {\it containment strategy} for an instance of \IV on input $(G,r,b)$ is a strategy $\sigma=(\sigma_1,\sigma_2,\dots,\sigma_m)$ for some finite integer $m$, with the property that at least one vertex is open at the start of time-step $m$, but no open vertex becomes infected during time-step $m$.
\end{definition}

Figure~\ref{fig:contain-grid1} illustrates a containment strategy for the infinite $2$-dimensional Cartesian grid: one vertex is inoculated during each of the first four time-steps, the containment strategy has length $5$, and containment does not occur until time-step $5$.  

For an instance of \IV on input $(G,r,b)$, suppose there is exactly one vertex, $v$, inoculated at time-step $t$, but no vertices are inoculated before or after $t$. The vaccinations stemming from the inoculation of $v$ will reach every vertex in the set $\{u \in V(G) : t+ d(v,u) \leq d(r,u)\}$ before the contagion.  More generally, let $\sigma$ be a strategy for \IV on input $(G,r,b)$.  If $v$ is inoculated during time-step $t$ following strategy $\sigma$, then the vertices in $\{u \in V(G) : t+ d(v,u) \leq d(r,u)\}$ are vaccinated as a direct consequence of the inoculation of $v$ during time-step $t$.  We formally name this set.  

\begin{definition}\label{defn:transmission} Let $\sigma$ be a 
strategy on input $(G,r,b)$ and $v$ be a vertex inoculated at time-step $t$. The {\it transmission set} of $v$ is denoted by $\mathcal{T}(v)$ and defined to be $$\mathcal{T}(v) = \{u \in V(G) : t+ d(v,u) \leq d(r,u)\}.$$
\end{definition} 

It is important to note that a vertex may be in the transmission set of more than one inoculated vertex. Transmission sets of inoculated vertices will be useful in capturing which vertices in a graph become vaccinated versus infected. 

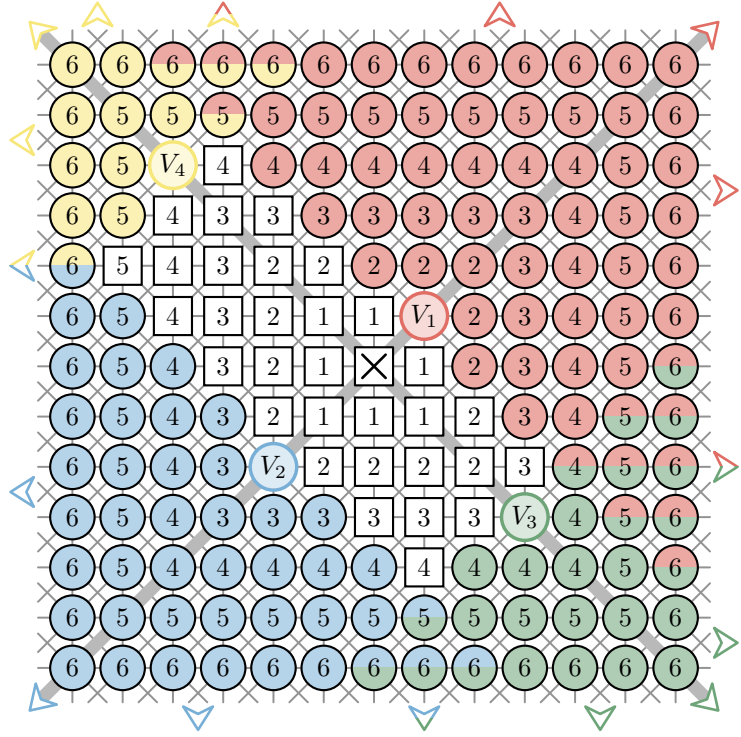
\begin{figure}[htb]
  \begin{center}
  \begin{tikzpicture}[line cap=round, line join=round,scale=0.7]

  \foreach \r in {1,...,13}{
    \foreach \c in {1,...,13}{
      \coordinate (g\r-\c) at ({(\c-1)*0.95},{-(\r-1)*0.95});
    }
  }

  \foreach \r in {1,...,13}{
    \foreach \c [count=\cc from 2] in {1,...,12}{
      \draw[gray!85, line width=0.75pt] (g\r-\c) -- (g\r-\cc);
    }
  }
  \foreach \c in {1,...,13}{
    \foreach \r [count=\rr from 2] in {1,...,12}{
      \draw[gray!85, line width=0.75pt] (g\r-\c) -- (g\rr-\c);
    }
  }
  \foreach \r [count=\rr from 2] in {1,...,12}{
    \foreach \c [count=\cc from 2] in {1,...,12}{
      \draw[gray!85, line width=0.75pt] (g\r-\c) -- (g\rr-\cc);
      \draw[gray!85, line width=0.75pt] (g\r-\cc) -- (g\rr-\c);
    }
  }

  \draw[line width=5pt, gray!55] (g7-7) -- (g1-13);
  \draw[line width=5pt, gray!55] (g7-7) -- (g13-1);
  \draw[line width=5pt, gray!55] (g7-7) -- (g1-1);
  \draw[line width=5pt, gray!55] (g7-7) -- (g13-13);

  \def\stub{0.65}
  \def\gap{0.5}

  \foreach \c in {1,...,13}{
    \draw[gray!85, line width=0.75pt] (g1-\c) -- ++(-\stub,\stub);
    \draw[gray!85, line width=0.75pt] (g1-\c) -- ++(0,\stub);
    \draw[gray!85, line width=0.75pt] (g1-\c) -- ++(\stub,\stub);

    \draw[gray!85, line width=0.75pt] (g13-\c) -- ++(-\stub,-\stub);
    \draw[gray!85, line width=0.75pt] (g13-\c) -- ++(0,-\stub);
    \draw[gray!85, line width=0.75pt] (g13-\c) -- ++(\stub,-\stub);

    \draw[gray!85, line width=0.75pt] (g\c-1) -- ++(-\stub,\stub);
    \draw[gray!85, line width=0.75pt] (g\c-1) -- ++(-\stub,0);
    \draw[gray!85, line width=0.75pt] (g\c-1) -- ++(-\stub,-\stub);

    \draw[gray!85, line width=0.75pt] (g\c-13) -- ++(\stub,\stub);
    \draw[gray!85, line width=0.75pt] (g\c-13) -- ++(\stub,0);
    \draw[gray!85, line width=0.75pt] (g\c-13) -- ++(\stub,-\stub);
  }

  \foreach \r in {1,...,13}{
    \foreach \c in {1,...,13}{
      \coordinate (h\r-\c) at ({(\c-1+0.5)*0.95},{-(\r-1)*0.95});
    }
  }
    \foreach \r in {1,...,13}{
    \foreach \c in {1,...,13}{
      \coordinate (h\r-\c) at ({(\c-1)*0.95},{-(\r-1+0.5)*0.95});
    }
  }

    \gridchevron{(0.5*0.95,1*0.95)}{90}{cYellow}
    \gridchevronsplit{(3*0.95,1*0.95)}{90}{cRed}{cYellow}
    \gridchevron{(8.5*0.95,1*0.95)}{90}{cRed}
  
    \gridchevron{(2.5*0.95,-13*0.95)}{270}{cBlue}
    \gridchevronsplit{(7*0.95,-13*0.95)}{270}{cBlue}{cGreen}
    \gridchevron{(10.5*0.95,-13*0.95)}{270}{cGreen}

    \gridchevron{(-1*0.95,-1.5*0.95)}{180}{cYellow}
    \gridchevronsplit{(-1*0.95,-4*0.95)}{180}{cYellow}{cBlue}
    \gridchevron{((-1*0.95,-8.5*0.95))}{180}{cBlue}

    \gridchevron{(13*0.95,-2.5*0.95)}{0}{cRed}
    \gridchevronsplit{(13*0.95,-8*0.95)}{0}{cRed}{cGreen}
  
    \gridchevron{(13*0.95,-11.5*0.95)}{0}{cGreen}

  \draw[line width=5pt, gray!55] (g1-1) -- ++(-0.65,0.65);
  \gridchevron{($(g1-1)+(-0.65,0.65)$)}{135}{cYellow}
  \draw[line width=5pt, gray!55] (g1-13) -- ++(0.65,0.65);
  \gridchevron{($(g1-13)+(0.65,0.65)$)}{45}{cRed}
  \draw[line width=5pt, gray!55] (g13-1) -- ++(-0.65,-0.65);
  \gridchevron{($(g13-1)+(-0.65,-0.65)$)}{225}{cBlue}
  \draw[line width=5pt, gray!55] (g13-13) -- ++(0.65,-0.65);
  \gridchevron{($(g13-13)+(0.65,-0.65)$)}{315}{cGreen}

  \foreach \r/\c/\type/\val in {
    1/1/yellow/6, 1/2/yellow/6, 1/3/RY/6, 1/4/RY/6, 1/5/RY/6, 1/6/red/6, 1/7/red/6, 1/8/red/6, 1/9/red/6, 1/10/red/6, 1/11/red/6, 1/12/red/6, 1/13/red/6,
    2/1/yellow/6, 2/2/yellow/5, 2/3/yellow/5, 2/4/RY/5, 2/5/red/5, 2/6/red/5, 2/7/red/5, 2/8/red/5, 2/9/red/5, 2/10/red/5, 2/11/red/5, 2/12/red/5, 2/13/red/6,
    3/1/yellow/6, 3/2/yellow/5, 3/3/Vfour/4, 3/4/S/4, 3/5/red/4, 3/6/red/4, 3/7/red/4, 3/8/red/4, 3/9/red/4, 3/10/red/4, 3/11/red/4, 3/12/red/5, 3/13/red/6,
    4/1/yellow/6, 4/2/yellow/5, 4/3/S/4, 4/4/S/3, 4/5/S/3, 4/6/red/3, 4/7/red/3, 4/8/red/3, 4/9/red/3, 4/10/red/3, 4/11/red/4, 4/12/red/5, 4/13/red/6,
    5/1/YB/6, 5/2/S/5, 5/3/S/4, 5/4/S/3, 5/5/S/2, 5/6/S/2, 5/7/red/2, 5/8/red/2, 5/9/red/2, 5/10/red/3, 5/11/red/4, 5/12/red/5, 5/13/red/6,
    6/1/blue/6, 6/2/blue/5, 6/3/S/4, 6/4/S/3, 6/5/S/2, 6/6/S/1, 6/7/S/1, 6/8/Vone/1, 6/9/red/2, 6/10/red/3, 6/11/red/4, 6/12/red/5, 6/13/red/6,
    7/1/blue/6, 7/2/blue/5, 7/3/blue/4, 7/4/S/3, 7/5/S/2, 7/6/S/1, 7/7/X/0, 7/8/S/1, 7/9/red/2, 7/10/red/3, 7/11/red/4, 7/12/red/5, 7/13/RG/6,
    8/1/blue/6, 8/2/blue/5, 8/3/blue/4, 8/4/blue/3, 8/5/S/2, 8/6/S/1, 8/7/S/1, 8/8/S/1, 8/9/S/2, 8/10/red/3, 8/11/red/4, 8/12/RG/5, 8/13/RG/6,
    9/1/blue/6, 9/2/blue/5, 9/3/blue/4, 9/4/blue/3, 9/5/Vtwo/2, 9/6/S/2, 9/7/S/2, 9/8/S/2, 9/9/S/2, 9/10/S/3, 9/11/RG/4, 9/12/RG/5, 9/13/RG/6,
    10/1/blue/6, 10/2/blue/5, 10/3/blue/4, 10/4/blue/3, 10/5/blue/3, 10/6/blue/3, 10/7/S/3, 10/8/S/3, 10/9/S/3, 10/10/Vthree/3, 10/11/green/4, 10/12/RG/5, 10/13/RG/6,
    11/1/blue/6, 11/2/blue/5, 11/3/blue/4, 11/4/blue/4, 11/5/blue/4, 11/6/blue/4, 11/7/blue/4, 11/8/S/4, 11/9/green/4, 11/10/green/4, 11/11/green/4, 11/12/green/5, 11/13/RG/6,
    12/1/blue/6, 12/2/blue/5, 12/3/blue/5, 12/4/blue/5, 12/5/blue/5, 12/6/blue/5, 12/7/blue/5, 12/8/BG/5, 12/9/green/5, 12/10/green/5, 12/11/green/5, 12/12/green/5, 12/13/green/6,
    13/1/blue/6, 13/2/blue/6, 13/3/blue/6, 13/4/blue/6, 13/5/blue/6, 13/6/blue/6, 13/7/BG/6, 13/8/BG/6, 13/9/BG/6, 13/10/green/6, 13/11/green/6, 13/12/green/6, 13/13/green/6
  }{
    \edef\pos{g\r-\c}

    \pgfkeys{/cell/.cd,
    yellow/.code={\gridcircle{\pos}{fillYellow}{\val}},
    red/.code={\gridcircle{\pos}{fillRed}{\val}},
    blue/.code={\gridcircle{\pos}{fillBlue}{\val}},
    green/.code={\gridcircle{\pos}{fillGreen}{\val}},
    RY/.code={\gridsplitcircle{\pos}{fillRed}{fillYellow}{\val}},
    YB/.code={\gridsplitcircle{\pos}{fillYellow}{fillBlue}{\val}},
    BG/.code={\gridsplitcircle{\pos}{fillBlue}{fillGreen}{\val}},
    RG/.code={\gridsplitcircle{\pos}{fillRed}{fillGreen}{\val}},
    S/.code={\gridsquare{\pos}{\val}},
    X/.code={\gridroot{\pos}},
    Vone/.code={\cellVone{\pos}{\val}},
    Vtwo/.code={\cellVtwo{\pos}{\val}},
    Vthree/.code={\cellVthree{\pos}{\val}},
    Vfour/.code={\cellVfour{\pos}{\val}}
    }
    \pgfkeys{/cell/\type}
  }

\end{tikzpicture}
\end{center}
\caption{Using budget $b=1$, a strategy for \IV on the infinite $2$-dimensional strong grid in which containment occurs at time-step $6$.} 
\label{fig:StrongGrid}
\end{figure} 
  
Figure~\ref{fig:StrongGrid} provides a containment strategy for the infinite $2$-dimensional strong grid: one vertex is inoculated during each of the first four time-steps, the containment strategy has length $6$, and containment does not occur until time-step $6$. The transmission set of each inoculated vertex is indicated by a different colour.
Observe from Figure~\ref{fig:StrongGrid} that the contagion spreads to only one vertex during time-step $5$.  If we had chosen that vertex to be inoculated during time-step $5$, then containment would have occurred during time-step $5$ instead of time-step $6$.  However, such a reduction in the number of time-steps to containment comes at the expense of a fifth inoculated vertex.  This motivates the next definition. 

\begin{definition}  
Let $\sigma=(\sigma_1,\dots,\sigma_m)$ be a containment strategy on input $(G,r,b)$. 
\begin{enumerate}
\item The {\it inoculation cost} of $\sigma$ is $\sum_{i=1}^m |\sigma_i|$.
\item The {\it inoculation number} on input $(G,r,b)$, denoted $\inoc(G,r,b)$, is the minimum inoculation cost, taken over all containment strategies on $(G,r,b)$.
\item The {\it containment time} on input $(G,r,b)$ is the minimum length of $\sigma$, taken over all  containment strategies of minimum inoculation cost.
\end{enumerate}
\end{definition}

To illustrate the above definitions, consider the containment strategy described by  Figure~\ref{fig:StrongGrid}.  For $b=1$, the inoculation number is at most $4$ and the containment time is at most $6$. As a second illustration, in Section~\ref{sec:alg}, we showed that for an interval graph $G$ and arbitrary vertex $r$, $\inoc(G,r,1) \leq 2$ and the containment time is at most 3.

\medskip

Finally, throughout this section, we will often use the notation $[n]$ to refer to the set of the first $n$ positive integers $\{1,2,\dots,n\}$.
 
\subsection{The infinite $n$-dimensional Cartesian grid}\label{sec:Cart}

In \bFire, the fire cannot be contained on the infinite $2$-dimensional Cartesian grid with budget $b=1$~\cite{moeller_2002}; and for $n \geq 3$, the fire cannot be contained on the infinite $n$-dimensional Cartesian grid with budget $b=2n-2$~\cite{develin_2004}.  In stark contrast, we will show that for \IV with budget $b=1$, the contagion can be contained on the infinite $n$-dimensional Cartesian grid for all $n$; and we additionally bound the containment time.  Let $G_n$ denote the infinite $n$-dimensional Cartesian grid; that is, the Cartesian product of $n$ infinite two-way paths. For simplicity of notation, we denote vertices in $G_n$ with vector notation.  Throughout this section we assume, without loss of generality, that the root of $G_n$ is the origin $\vec{0} = (0,0,\dots,0)$. 
Observe that for any two vertices $\vec{u} = (u_1,u_2,\dots,u_n)$ and $\vec{w} = (w_1,w_2,\dots,w_n)$ in $V(G_n)$, the distance between $\vec{u}$ and $\vec{w}$ is $$d(\vec{u},\vec{w}) = \sum_{i=1}^n |u_i-w_i|.$$ For $i \in [n]$, let $J^+_i$ ($J^-_i$) denote the set of $n$-tuples in which the $i^{th}$ entry is a positive (negative) integer and all other entries are zero.  We define the {\it axes} of $G_n$ to be $J^+_1,J^+_2,\dots,J^+_n,J^-_1,J^-_2,\dots,J^-_n$.

\begin{theorem} \label{lem:T_i}
Let $\sigma$ be a containment strategy on input $(G_2,\vec{0},1)$.  The transmission set of each vertex of $\sigma$ contains  vertices from at most one axis of $G_2$.
\end{theorem}

\begin{proof} Let $\vec{u} = (u_1,u_2)$ be a vertex of $\sigma$, inoculated at some time-step $t \geq 1$. Without loss of generality, $u_1>0$ and $u_1 \geq u_2 \geq 0$.  Let $\vec{w} = (0,w_2) \in J_2^+$.  Then $t+d(\vec{u},\vec{w}) = t+u_1+|u_2-w_2| > w_2 = d(\vec{0},\vec{w})$.  Thus, $\vec{w} \not\in \mathcal{T}(\vec{u})$ and $\mathcal{T}(\vec{u})$ contains no vertices of $J^+_2$.  A similar argument shows $\mathcal{T}(\vec{u})$ contains no vertices of $J_1^-$ or $J_2^-$.  The corresponding argument can be applied to inoculated vertices in other quadrants of $G_2$.
\end{proof} 

Corollary~\ref{cor:axes} follows immediately as it can be applied to each pair of the $2n$ axes in $G_n$.  

\begin{corollary}\label{cor:axes} Let $\sigma$ be a containment strategy on input $(G_n,\vec{0},1)$ for $n \geq 2$.   The transmission set of each vertex of $\sigma$ contains vertices from at most one axis of $G_n$.  
\end{corollary}

Corollary \ref{cor:axes} provides the lower bound for the minimum cost of a containment strategy on $(G_n,\vec{0},1)$, which we use in the proof of the next result.

\begin{theorem}\label{thm:uppercart} For $n \geq 2$, $\inoc(G_n,\vec{0},1)=2n$ and the containment time on $(G_n,\vec{0},1)$ is at most $n^2+1$. 
\end{theorem}

\begin{proof} 
  Let $\sigma=(\sigma_1,\dots,\sigma_m)$ be a minimum cost containment strategy on input $(G_n,\vec{0},1)$.  
  Since $G_n$ has $2n$ axes and each transmission set contains vertices from at most one axis, the lower bound follows from Corollary~\ref{cor:axes}.  
  To prove $\inoc(G_n,\vec{0},1) \leq 2n$, we next provide a sequence $\sigma$ of cost $2n$ and prove it is a containment strategy. 
  For each $i \in [n]$, let $\vec{x}_{2i-1} \in J^+_i$ with $i^{th}$ coordinate $x_i = 2i-1$ and $\vec{y}_{2i} \in J^-_i$ with $i^{th}$ coordinate $y_i = -2i$.  Set $\sigma_{2i-1}  = \{ \vec{x}_{2i-1}\}$ and $\sigma_{2i} = \{\vec{y}_{2i}\}$ for $i \in \{1,2,\dots,n\}$; and set $\sigma_i = \emptyset$ for $i \in \{n+1, \dots, m\}$. 
  Observe that $d(\vec{x}_{2i-1},\vec{0})=2i-1$, but since $\vec{x}_{2i-1}$ is inoculated before the contagion spreads during time-step $2i-1$, $\vec{x}_{2i-1}$ is open at the beginning of time-step $2i-1$ and will be inoculated.  
  Similarly, $\vec{y}_{2i}$ is open at the beginning of time-step $2i$. 
  Thus, to prove $\sigma$ is a containment strategy, it suffices to prove there is some finite time-step $m$ during which no open vertex becomes infected.
  For a contradiction, assume that at every finite time-step, the contagion spreads to an open vertex. Then, for some minimum time-step $t$, there is a $j \in [n]$ for which
  \begin{enumerate} 
 \item[(1)] a vertex $\vec{w} = (w_1,\dots,w_n)$ with $w_j= 2j-1$ becomes infected at some time-step $t \geq 2j-1$; or 
 \item[(2)] a vertex $\vec{w} = (w_1,\dots,w_n)$ with $w_j= -2j$  becomes infected at some time-step $t \geq 2j$.
 \end{enumerate}
 
First, suppose (1) occurs. Observe that $$t \geq d(\vec{0},\vec{w}) = \sum_{i=1}^n |w_i| = \big(\sum_{i=1}^n |w_i|\big) - (2j-1) + (2j-1) = d(\vec{x}_{2j-1},\vec{w})+2j-1.$$  Therefore, the earliest time-step that $\vec{w}$ can become infected is time-step $d(\vec{x}_{2j-1},\vec{w})+2j-1$.  By definition, open vertices that are adjacent to vaccinated vertices, become vaccinated at the {\it beginning} of a time-step, whereas open vertices that are adjacent to infected vertices, become infected {\it later} in a time-step.  Thus, $\vec{w}$ will become vaccinated during time-step $d(\vec{x}_{2j-1},\vec{w})+2j-1$ and (1) cannot occur.

Second, suppose (2) occurs. Observe that $$t \geq d(\vec{0},\vec{w}) = \sum_{i=1}^n |w_i| = \big( \sum_{i=1}^n |w_i| \big) - (2j) + (2j) = d(\vec{y}_{2i},\vec{w})+2j.$$ 
By a similar argument to (1),  $\vec{w}$ will become vaccinated during time-step $d(\vec{y}_{2j},\vec{w})+2j$ and (2) cannot occur.

 Therefore, $\sigma$ is a containment strategy with inoculation cost $2n$ and containment time $m$.   Furthermore, if some vertex $\vec{w}=(w_1,w_2,\dots,w_n)$ becomes infected during the process, then $-2j < w_j < 2j-1$ for all $j \in \{1,2,\dots,n\}$. 
Assume $\vec{w}$ becomes infected during time-step $m-1$.  Since the inoculated vertices are all on axes, $m-1 = d(\vec{0},\vec{w})$ and $$ m-1 = d(\vec{0},\vec{w}) = \sum_{i=1}^n |w_i| \leq \sum_{i=1}^n |-2i+1| = \sum_{i=1}^n (2i-1) = n^2.$$
\end{proof}

 Corollary~\ref{cor:axes} implies that $\inoc(G_n,\vec{0},b) \geq 2n$ for all $b \geq 1$.  Observe that $\inoc(G,r,b) \leq \inoc(G,r,1)$ for any graph $G$, fixed root $r$, and $b>1$. Corollary~\ref{cor:cartb} follows immediately.

\begin{corollary}\label{cor:cartb} For $b\geq 1$, $\inoc(G_n,\vec{0},b) =2n$. \end{corollary} 

With respect to the containment time on $(G_n,\vec{0},1)$, we conjecture that the bound in Theorem~\ref{thm:uppercart} is best possible.  Since $\inoc(G_n,\vec{0},b)=2n$ for all $b \geq 1$, the containment time on $(G_n,\vec{0},b)$ is bounded above by the containment time for $(G_n,\vec{0},1)$. 

\begin{ques}For $b\geq 1$, what is the containment time for $(G_n,\vec{0},b)$?
\end{ques}

 \subsection{The infinite $n$-dimensional strong grid}

The \emph{strong product} of two graphs $H$ and $H'$, written $H \boxtimes H'$, has vertex set $V(H) \times V(H')$. Vertices $(a, b)$ and $(c, d)$ in $H \boxtimes H'$ are adjacent if (1) $a = c \in V(H)$ and
$(b,d) \in E(H')$, or (2) $(a,c) \in E(H)$ and $b = d \in V(H')$, or (3) $(a,c) \in E(H)$ and $(b,d) \in E(H')$. Let $H_n$ denote the infinite $n$-dimensional strong grid; that is, the strong product of $n$ infinite two-way paths. We continue to use vector notation, and as in Section \ref{sec:Cart}, we assume, without loss of generality, that the root of $H_n$ is the origin $\vec{0}$. We again focus firstly on minimising budget and secondly in minimising the number of inoculations. That is, we show that the contagion can be contained with a budget of only $b=1$ and then determine the minimum cost of a containment strategy for $(H_n,\vec{0},1)$.
Note that for any $\vec{u} = (u_1,u_2,\dots,u_n)$ and $\vec{w} = (w_1,w_2,\dots,w_n)$ in $V(H_n)$, the distance between $\vec{u}$ and $\vec{w}$ is $$d(\vec{u},\vec{w}) = \max_{i \in [n]}|u_i-w_i|.$$ 

Let $A_n$ denote the set of all vectors of length $n$ whose entries are from $\{-1,1\}$. 

The Hadamard product is a binary operation that takes two vectors of the same length and returns a vector of the multiplied corresponding elements.  For example, the Hadamard product of $(a,b,c)$ and $(x,y,z)$ is $(a,b,c) \circ (x,y,z) = (ax, by, cz)$.  Note that we are using vector notation for convenience, but $\vec{a} \circ \vec{x}$ and $x \cdot \vec{a}$ represent {\it vertices} in the graph.

For each $\vec{a} \in A_n$, define the {\it orthant of $\vec{a}$} to be\footnote{A similar definition is used in matroid theory.  In particular, $\vec{v}$ \emph{conforms} to $\vec{a}$ if each non-zero coordinate of $\vec{v}$ has the same sign as the corresponding coordinate in $\vec{a}$.  Then $\mathcal{O}_{\vec{a}}$ is defined to be the set of all non-zero (integer-valued) vectors that conform to $\vec{a}$.} $$\mathcal{O}_{\vec{a}} = \{ \vec{a} \circ \vec{x} ~:~ x_i \in \mathbb{N} \cup \{0\}~\forall i \in [n] \text{~and~} \vec{x} \neq \vec{0}\hspace{0.5mm}\}.$$  Thus, $\vec{a}$ dictates the signs of the coordinates of the vertices in $\mathcal{O}_{\vec{a}}$. Define the {\it set of all orthants of $H_n$} to be $$\mathcal{O} = \bigcup_{\vec{a} \in A_n} \mathcal{O}_{\vec{a}}$$ and note that $\mathcal{O} = V(H_n) \backslash \{\vec{0}\}$. For each $\vec{a} \in A_n$, the {\it diagonal orthant of $\vec{a}$} is the set $\{ x \cdot\vec{a} ~:~ x \in \mathbb{N}\}$.

\bigskip

To illustrate the above definitions, consider the infinite $2$-dimensional strong grid $H_2$ and observe:
\begin{itemize}
\item $A_2 = \{(1,1),(-1,-1),(-1,1),(1,-1)\}$;
\item  $\mathcal{O}_{(1,1)} = \{(u,v) \in \mathbb{Z} : u,v \geq 0 \text{~and at least one of $u,v$ is non-zero} \}$;  
\item  $\mathcal{O}_{(-1,-1)} = \{(u,v) \in \mathbb{Z} : u,v \leq 0 \text{~and at least one of $u,v$ is non-zero} \}$;  
\item  $\mathcal{O}_{(-1,1)} = \{(u,v) \in \mathbb{Z} : u \leq 0, v \geq 0, \text{~and at least one of $u,v$ is non-zero} \}$;  
\item  $\mathcal{O}_{(1,-1)} = \{(u,v) \in \mathbb{Z} : u \geq 0, v \leq 0, \text{~and at least one of $u,v$ is non-zero} \}$.  
\end{itemize}  Note that vertices of the form $(u,0)$ or $(0,v)$ for $u, v \neq 0$ 
are in the orthant of two different vertices of $A_2$.  For example, for $u>0$, $(u,0)$  
is in both $\mathcal{O}_{(1,1)}$ and $\mathcal{O}_{(1,-1)}$. 
 For $n=2$, the diagonal orthants are highlighted with thick grey lines in Figure~\ref{fig:StrongGrid}.  Figure~\ref{fig:StrongGrid} also illustrates how the minimum cost of a containment strategy, that is, the inoculation cost, for $(H_n,\vec{0},1)$ is at most four. 

The following lemma will help to determine the minimum cost of a containment strategy for the infinite $n$-dimensional strong grid in an analogous fashion to how Theorem \ref{lem:T_i} helped to prove the minimum cost of a containment strategy for the infinite $n$-dimensional Cartesian grid in Section \ref{sec:Cart}. 

\begin{lemma} \label{lem:oneOrth} For $n \geq 2$, consider a containment strategy $\sigma$ on input $(H_n,\vec{0},1)$.  Each transmission set contains vertices from at most one diagonal orthant of $H_n$.
\end{lemma}

\begin{proof} Let $\sigma$ be a containment strategy on input $(H_n,\vec{0},1)$.  Denote by $D_{\vec{a}}$ and $D_{\vec{b}}$, two distinct diagonal orthants of $H_n$. As a contradiction, suppose $\vec{u} = (u_1,u_2,\dots,u_n)$ is inoculated during time-step $t$ and $\mathcal{T}(\vec{u})$ contains both vertices $\vec{x} = (a_1x,a_2x,\dots,a_nx) \in D_{\vec{a}}$ and $\vec{y} = (b_1y,b_2y,\dots,b_ny) \in D_{\vec{b}}$ for some positive integers $x$ and $y$. Without loss of generality, assume $x < y$. 

Since $\vec{y} \in \mathcal{T}(\vec{u})$ and $\vec{u}$ is inoculated during time-step $t$, by Definition~\ref{defn:transmission} $$t+\max_{i \in [n]} |u_i-b_iy| =t+d(\vec{u},\vec{y}) \leq  d(\vec{0},\vec{y}) = y.$$ Since $t \geq 1$, it follows that $\max_{i \in [n]}|u_i-b_iy| < y$.  We similarly deduce that $\max_{j \in [n]} |u_j - a_jx| < x$.   As $D_{\vec{a}} \ne D_{\vec{b}}$, there must be some $i \in [n]$ such that $b_i \ne a_i$. Without loss of generality, suppose $b_1 \ne a_1$. Then $b_1=-a_1$ as $b_1,a_1 \in \{-1,1\}$.  It follows that \begin{equation}\label{TheFirstEquation} \text{(a)}~~|u_1-a_1x|<x ~~~~\text{and~~~~ (b)}~~|u_1+a_1y|<y.\end{equation}

If $a_1=-1$ then Inequality~(\ref{TheFirstEquation})(a) simplifies to $|u_1+x|<x$.  As $x>0$, this implies $u_1<0$, which contradicts Inequality~(\ref{TheFirstEquation})(b).  Similarly, if $a_1=1$ then Inequality~(\ref{TheFirstEquation})(b) simplifies to $|u_1+y|<y$.  As $y>0$, this implies $u_1<0$, which contradicts Inequality~(\ref{TheFirstEquation})(a).
\end{proof}

As there are $2^n$ distinct diagonal orthants, Lemma \ref{lem:oneOrth} implies $\inoc(G) \geq 2^n$.\medskip

Recall that $A_n$ denotes the set of length $n$ vectors, whose coordinates are in $\{-1,1\}$.  As there is a bijection between the vectors of $A_n$ and the set of bit-strings of length $n$, we arbitrarily label the vectors in $A_n$ as $\vec{a}_1,\vec{a}_2,\dots,\vec{a}_{2^n}$. In the next result, we present a sequence $\sigma$ and show that it bounds from above, the coordinates of any infected vertex.  We then prove that $\sigma$ is a containment strategy of length $2^n$ for $(H_n,\vec{0},1)$ and provide an upper bound for the containment time. 

\begin{lemma}\label{lem: strong containment} Let $A_n = \{\vec{a}_1,\vec{a}_2,\dots,\vec{a}_{2^n}\}$ and $\sigma = (\sigma_1,\dots,\sigma_m)$ where $$\sigma_i = \begin{cases} \{\vec{a}_i \cdot i\} & \text{~if~} i \in \{1,\dots,2^n\}, \\  \emptyset & \text{~if~} i \in \{2^n+1,\dots,m\}.\end{cases}$$  Let  
$\vec{w} \neq \vec{0}$
be a vertex infected by the end of time-step $m$.  Then for some $x \in \{1,2,\dots,2^n\}$, 
$\vec{w} \in \mathcal{O}_{\vec{a}_x}$.  Let $\vec{a}_y \in A_n$ such that $\vec{a}_y$ differs from $\vec{a}_x$ only in coordinate $j$ for some $j \in [n]$.  Then $|w_i| < x+y$ for all $i \neq j$.
\end{lemma}

\begin{proof} We express $\vec{w}$ 
as $\vec{w} = (a_1w_1,a_2w_2,\dots,a_nw_n)$ where each $w_i$ is a nonnegative integer (at least one being positive) and each $a_i \in \{-1,1\}$.  Then $\vec{a}_x = (a_1,a_2,\dots,a_n) \in A_n$ and 
$\vec{w} \in \mathcal{O}_{\vec{a}_x}$.  Let $\vec{a}_y=(b_1,b_2,\dots,b_n)$ where $\vec{a}_y$ differs from $\vec{a}_x$ only in coordinate $j$.  Without loss of generality, let $j=1$; and in particular, let $a_1=1$ and $b_1=-1$.  Following $\sigma$, note that vertices $\vec{a}_x \cdot x$ and $\vec{a}_y \cdot y$ are inoculated at time-steps $x$ and $y$, respectively.  Without loss of generality, suppose $x > y$.

We aim to prove $w_i < x+y$ for $i \neq j$.  For a contradiction, we  assume $w_i \geq x+y$ for some $i \neq 1$. However, if such a vertex $\vec{w}$ is infected at time-step $t$, then there is an earlier time-step $t'<t$ in which $\vec{v}$ is infected, where $\vec{v}=(a_1v_1,\dots,a_nv_n)$ with $v_1\geq 0$, $0<v_j \leq x+y$ for $j \neq 1$, and $v_j=x+y$ for some $j \in [n]$.
  
Then $\vec{v}= (b_1v_1',b_2v_2',\dots,b_nv_n')$ where $b_1 = -a_1$, $v_1'=-v_1$;  and for $i>1$, $b_i = a_i$ and $v_i'=v_i$.  Since $\vec{v} \not\in \mathcal{T}(\vec{a}_y\cdot y)$,  $$x+y = \max_{i \in [n]} v_i = d(\vec{0},\vec{v})  < y+d(\vec{a}_y \cdot y, \vec{v}) = y+\max_{i \in [n]} |b_iy-b_iv_i'| = y+\max_{i \in [n]} |y-v_i'|$$ $$\Longleftrightarrow~ x < \max_{i \in [n]}\big\{y+v_1,|y-v_2|,|y-v_3|,\dots,|y-v_n|\big\}.$$  

\noindent \underline{Case 1:} Suppose $\max_{i \in [n]} |y-v_i'| = |y-v_r|$ for some $r>1$.  Then $x<|y-v_r|$.  If $v_r \geq y$ then $v_r-y > x ~\Leftrightarrow~ v_r > x+y$, a contradiction.  Thus, $y > v_r$ and furthermore,  $$\max_{i \in [n]} |y-v_i'| = |y-v_r| \geq y+v_1 ~\Longleftrightarrow~ y-v_r \geq y+v_1 ~\Longleftrightarrow~ v_1+v_r \leq 0,$$ which implies $v_1=v_r=0$. But then $x<|y-v_r| ~\Leftrightarrow~ x<y$, which contradicts $x > y$.  Thus, Case 1 cannot occur.  \medskip

\noindent \underline{Case 2:} Suppose $\max_{i \in [n]} |y-v_i'| = y+v_1$.  Then $v_1 > x-y$.  Since Case 1 does not occur, $y+v_1 >  |y-v_i|$ for all $i>1$.  We consider $|y-v_i|$ for each $i>1$. 
\begin{itemize}
\item If $v_i \leq y$ then $v_i \leq y < x+y$.
\item If $v_i > y$ then $y+v_1 > |y-v_i| = v_i - y ~~\Longleftrightarrow~~ 2y > v_i-v_1$. 
\begin{itemize} 
\item If $v_i \geq v_1$ then $0 \leq v_i -v_1 < 2y < x+y$, which implies $v_i < x+y$. 
\item  If $v_i<v_1$ then $v_i < v_1 \leq x+y$.  
\end{itemize}
In all situations, we find $v_i < x+y$ for all $i>1$, a contradiction.
\end{itemize} 
For Case 2, there is no situation in which $v_i = x+y$ for any $i>1$.\medskip

Thus, if $\vec{w} \neq \vec{0}$ is infected by the end of time-step $m$, then $w_i < x+y$ for all $i \neq j$.  A similar argument shows that $w_i > -(x+y)$ for all $i \neq j$.  Thus, $|w_i| < x+y$ for all $i \neq j$.
\end{proof}

We next introduce a parameter that will provide an upper bound for the containment time for $(H_n,\vec{0},1)$. Let $G$ be a graph and $f$ be an ordering of the vertices of $G$. The {\it strength of $f$}, denoted $\strength_f(G)$, is \begin{equation}\label{eqnstr1} \strength_f(G) = \max \big|\{ f(u)+f(v) ~:~ uv \in E(G)\}\big|\end{equation} and the {\it strength} of $G$, denoted $\strength(G)$, is \begin{equation}\label{eqnstr2}str(G) = \min \big|\{ \strength_f(G)~:~ f \text{~is a linear ordering of~} G\}\big|.\end{equation}  The parameters of~(\ref{eqnstr1}) and~(\ref{eqnstr2}) were first defined in~\cite{I.2018}, motivated by super edge-magic labelings and the super magic strength of a graph.  The strength of a graph is a generalisation of the concept of super magic strength.  See~\cite{Av,I.2018} for more on super edge-magic labelings and super magic strength.  Table~\ref{table:strength} details exact values and best known upper bounds for $\strength(Q_n)$. 
We make use of the results in Table~\ref{table:strength} to bound the containment time for $(H_n,\vec{0},1)$.

 \begin{theorem}\label{thm:upperstrongQn} For $n \geq 2$,  $\inoc(H_n,\vec{0},1)=2^n$ and the containment time is at most $str(Q_n)$.
\end{theorem}

\begin{proof} 
That $\inoc(H_n,\vec{0},1) \geq 2^n$ follows from Lemma~\ref{lem:oneOrth}.  To prove $\inoc(H_n,\vec{0},1) \leq 2^n$ and provide an upper bound on the containment time, we first define a sequence $\sigma = (\sigma_1,\dots,\sigma_m)$:  let $\sigma_t = \{\vec{a}_t \cdot t\}$ for $t \in \{1,\dots,2^n\}$ and $\sigma_t = \emptyset$ for $t \in \{2^n+1,\dots,m\}$.

Recall that the elements of $A_n$ were given an arbitrary ordering $\vec{a}_1,\vec{a}_2,\dots,\vec{a}_{2^n}$.  We apply Lemma~\ref{lem: strong containment} to each pair of elements in $A_n$ that differ in exactly one coordinate.  As a result,  the largest coordinate of any infected vertex is at most \begin{equation}\label{eqn2}\max \big|\{i+j-1 : i,j \in \{1,2,\dots,2^n\} \text{~and, $\vec{a}_i$ and $\vec{a}_j$ differ in exactly one coordinate}\}\big|.\end{equation}  If $m$ is larger than the quantity in~(\ref{eqn2}), the contagion is contained, and so  $\sigma$ is a containment strategy for $(H_n,\vec{0},1)$. Furthermore, the quantity in~(\ref{eqn2}) provides us with an upper bound for the containment time of $(H_n,\vec{0},1)$. 

Each element of $A_n$ corresponds to a distinct vertex in the hypercube $Q_n$.  So $\vec{a}_i$ and $\vec{a}_j$ differ in exactly one coordinate if and only if corresponding vertices $x_i$ and $x_j$ in $Q_n$ are adjacent.  Given an arbitrary ordering of the elements of $A_n$, any bijection $f : \{x_1,x_2,\dots,x_{2^n}\} \rightarrow \{1,2,\dots,2^n\}$ will correspond to the following bound on the containment time: $$\max \big| \{f(u)+f(v): uv \in E(Q_n)\}\big| = \max \big| \{f(u)+f(v) : uv \in E(Q_n)\}\big| = str_f(Q_n).$$  Thus, the containment time for $(H_n,\vec{0},1)$ is at most $$\strength(Q_n) = \min \big| \{\strength_f(Q_n) : f \text{~is a linear ordering of~} Q_n\}\big|.$$

\end{proof}

The containment time for $(H_n,\vec{0},1)$ is therefore bounded by the values given in Table~\ref{table:strength}.

\begin{table}[htbp]
\[ 
\begin{tabular}{cc}\begin{tabular}{|c|c|c|}
\hline
$n$ & $\strength(Q_n)$ & \text{comment} \\
\hline
1 & = 3 & due to \cite{I.2018}\\
2 & = 6 & due to \cite{I.2018} \\
3 & = 11 & due to \cite{G.L.S.2021}\\
4 & = 21 & due to \cite{G.L.S.2021}\\
5 & = 40 & due to \cite{G.L.S.2021}\\
6 & $\leq 78$ & due to \cite{Paper2} \\
7 & $\leq 152$ & due to \cite{Paper2} \\ 
\hline
\end{tabular} 
&
\begin{tabular}{|c|c|c|}
\hline
$n$ & $\strength(Q_n)$ & \text{comment} \\
\hline
8 & $\leq 300$ & due to \cite{Paper2} \\
9 & $\leq 591$ &  due to \cite{Paper2} \\
10 & $\leq 1173$ & due to \cite{Paper2} \\
11 & $\leq 2323$ & due to \cite{Paper2} \\
12 & $\leq 4623$ & due to \cite{Paper2} \\
13 & $\leq 9181$ & due to \cite{Paper2} \\
14+ & $\leq 2^n+2^{n-3}+28$ & due to \cite{Paper2} \\ 
\hline
\end{tabular} 
\end{tabular}
\] 
\caption{Known values and best current upper bounds for $\strength(Q_n)$.}
\label{table:strength}
\end{table}

Similar to the $n$-dimensional Cartesian grid, the inoculation number of $H_n$ is $2^n$ for all $b \geq 1$.

\begin{corollary}\label{cor:b} For $b\geq1$, $\inoc(H_n,\vec{0},b)=2^n$. \end{corollary} 

The lower bound in Corollary~\ref{cor:b} follows from Lemma~\ref{lem:oneOrth} and the upper bound follows because we have $\inoc(G,r,b) \leq \inoc(G,r,1)$ for any graph $G$, fixed root $r$, and $b>1$.  The containment time on $(H_n,\vec{0},b)$ is bounded above by the containment time on $(H_n,\vec{0},1)$, but the actual values remain unknown.

\begin{ques}For $b\geq 1$, what is the containment time for $(H_n,\vec{0},b)$?
\end{ques}

\subsection{The inoculation number of subgraphs}\label{sec:subgraphs}

We conclude the paper with a result about subgraphs. A natural question is:  {\it if $H$ is an induced subgraph of $G$, then is the inoculation number of $H$ at most that of $G$?}  In general, the answer is no, even when the subgraph is connected or isometric as shown by Theorem~\ref{thm:example} below.

Let $G^*$ be the graph shown on the left-side of Figure~\ref{fig:IsoEx} where the ellipses indicate paths of infinite length. The subgraphs induced by the deletion of center $c$ from $G^*$ are labelled $B_1$, $B_2$, and $B_3$; and observe that $B_1 \cong B_2 \cong B_3$.  Let $H^*$ be the subgraph illustrated on the right-side of Figure~\ref{fig:IsoEx} and observe it is the subgraph of $G^*$, induced by the deletion of three particular vertices.  

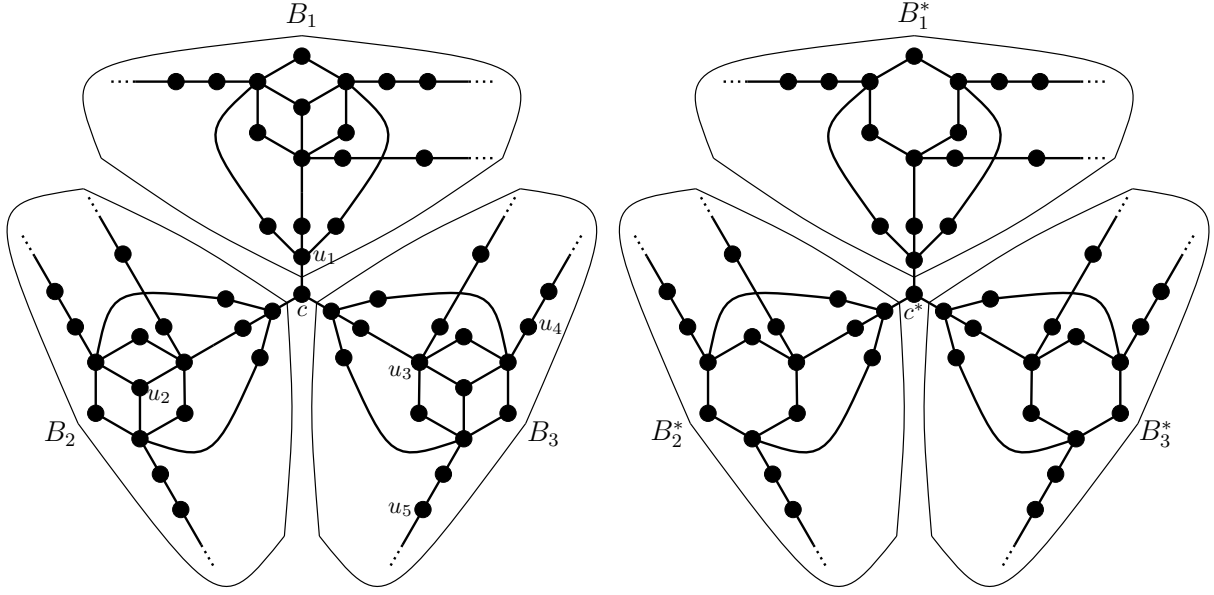
\begin{figure}[H]
\centering
\scalebox{0.9}{
   \begin{tikzpicture}

        \draw[line width = 1] (0,0) -- (0,1.5);
        \draw[line width = 1] (0,0.5) -- (-0.5,1);
        \draw[line width = 1] (0,0.5) -- (0.5,1);

        \draw[line width = 1] (0,2) -- (0.65,2.375) -- (0.65,3.125) -- (0,3.5) -- (-0.65,3.125) -- (-0.65,2.375) -- (0,2);

        \draw[line width = 1] (0,2.75) -- (0.65,3.125);
        \draw[line width = 1] (0,2.75) -- (-0.65,3.125);
        \draw[line width = 1] (0,2.75) -- (0,1.5);

        \draw[line width = 1] (-0.5,1) .. controls (-1.5,2.375) .. (-0.65,3.125);
        \draw[line width = 1] (0.5,1) .. controls (1.5,2.375) .. (0.65,3.125);


        \draw[line width = 1] (-0.65,3.125) --(-2.45,3.125);
        \draw[line width = 1,dotted] (-2.45,3.125) --(-2.85,3.125);

        \draw[line width = 1] (0.65,3.125) --(2.45,3.125);
        \draw[line width = 1,dotted] (2.45,3.125) -- (2.85,3.125);

        \draw[line width = 1] (0,2) --(2.45,2);
        \draw[line width = 1,dotted] (2.45,2) --(2.85,2);

        \draw[fill] (-1.25,3.125) circle [radius=0.12];
        \draw[fill] (-1.85,3.125) circle [radius=0.12];

        \draw[fill] (1.25,3.125) circle [radius=0.12];
        \draw[fill] (1.85,3.125) circle [radius=0.12];

        \draw[fill] (0.6,2) circle [radius=0.12];
        \draw[fill] (1.8,2) circle [radius=0.12];
        
        \draw[fill] (0,0.55) circle [radius=0.12];
        \node at (0.34,0.54) {$u_1$};
        \draw[fill] (0.5,1) circle [radius=0.12];
        \draw[fill] (-0.5,1) circle [radius=0.12];
        \draw[fill] (0,1) circle [radius=0.12];
        \draw[fill] (0,2) circle [radius=0.12];
        \draw[fill] (0.65,2.375) circle [radius=0.12];
        \draw[fill] (-0.65,2.375) circle [radius=0.12];
        \draw[fill] (0.65,3.125) circle [radius=0.12];
        \draw[fill] (-0.65,3.125) circle [radius=0.12];
        \draw[fill] (0,3.5) circle [radius=0.12];
        \draw[fill] (0,2.75) circle [radius=0.12];

        \draw[line width = 0.5pt] (0,3.8) .. controls (3.5,3.4) .. (2.95,2);
        \draw[line width = 0.5pt] (2.95,2) .. controls (1.5,1) .. (0,0.25);
        \draw[line width = 0.5pt] (0,3.8) .. controls (-3.5,3.4) .. (-2.95,2);
        \draw[line width = 0.5pt] (-2.95,2) .. controls (-1.5,1) .. (0,0.25);

        \node at (0,4.1) {\large$B_1$};


        \draw[line width = 1] (0,0) -- (-1.3,-0.75);
        \draw[line width = 1] (-0.433,-0.25) -- (-0.616,-0.933);
        \draw[line width = 1] (-0.433,-0.25) -- (-1.12,-0.07);

        \draw[line width = 1] (-1.73,-1) -- (-2.38,-0.625) -- (-3.03,-1) -- (-3.03,-1.75) -- (-2.38,-2.125) -- (-1.72,-1.75) -- (-1.73,-1);

        \draw[line width = 1] (-2.38,-1.375) -- (-3.03,-1);
        \draw[line width = 1] (-2.38,-1.375) -- (-2.38,-2.125);
        \draw[line width = 1] (-2.38,-1.375) -- (-1.3,-0.75);

        \draw[line width = 1] (-0.616,-0.933) .. controls (-1.31,-2.49) .. (-2.38,-2.125);
        \draw[line width = 1] (-1.12,-0.07) .. controls (-2.81,0.11) .. (-3.03,-1);


        \draw[line width = 1] (-2.38,-2.125) --(-1.48,-3.684);
        \draw[line width = 1,dotted] (-1.48,-3.684) --(-1.28,-4.03);
        
        \draw[line width = 1] (-3.03,-1) --(-3.93,0.559);
        \draw[line width = 1,dotted] (-3.93,0.559) -- (-4.13,0.9);

        \draw[line width = 1] (-1.73,-1) --(-2.96,1.12);
        \draw[line width = 1,dotted] (-2.96,1.12) --(-3.16,1.47);

        \draw[fill] (-2.08,-2.645) circle [radius=0.12];
        \draw[fill] (-1.78,-3.165) circle [radius=0.12];

        \draw[fill] (-3.33,-0.48) circle [radius=0.12];
        \draw[fill] (-3.63,0.04) circle [radius=0.12];

        \draw[fill] (-2.03,-0.48) circle [radius=0.12];
        \draw[fill] (-2.63,0.59) circle [radius=0.12];
        
        \draw[fill] (-0.433,-0.25) circle [radius=0.12];
        \draw[fill] (-1.12,-0.07) circle [radius=0.12];
        \draw[fill] (-0.616,-0.933) circle [radius=0.12];
        \draw[fill] (-.866,-0.5) circle [radius=0.12];
        \draw[fill] (-1.73,-1) circle [radius=0.12];
        \draw[fill] (-2.38,-0.625) circle [radius=0.12];
        \draw[fill] (-1.72,-1.75) circle [radius=0.12];
        \draw[fill] (-3.03,-1) circle [radius=0.12];
        \draw[fill] (-2.38,-2.125) circle [radius=0.12];
        \draw[fill] (-3.03,-1.75) circle [radius=0.12];
        \draw[fill] (-2.38,-1.375) circle [radius=0.12]; 
        \node at (-2.1,-1.5) {$u_2$};

        \draw[line width = 0.5pt] (-3.29,-1.9) .. controls (-4.694,1.33) .. (-3.21,1.55);
        \draw[line width = 0.5pt] (-3.21,1.55) .. controls (-1.616,0.8) .. (-0.2165,-0.125);
        \draw[line width = 0.5pt] (-3.29,-1.9) .. controls (-1.194,-4.73) .. (-0.257,-3.555);
        \draw[line width = 0.5pt] (-0.257,-3.555) .. controls (-0.116,-1.8) .. (-0.2165,-0.125);

        \node at (-3.551,-2.05) {\large$B_2$};


        \draw[line width = 1] (0,0) -- (1.3,-0.75);
        \draw[line width = 1] (0.433,-0.25) -- (0.616,-0.933);
        \draw[line width = 1] (0.433,-0.25) -- (1.12,-0.07);

        \draw[line width = 1] (1.73,-1) -- (2.38,-0.625) -- (3.03,-1) -- (3.03,-1.75) -- (2.38,-2.125) -- (1.72,-1.75) -- (1.73,-1);

        \draw[line width = 1] (2.38,-1.375) -- (3.03,-1);
        \draw[line width = 1] (2.38,-1.375) -- (2.38,-2.125);
        \draw[line width = 1] (2.38,-1.375) -- (1.3,-0.75);

        \draw[line width = 1] (0.616,-0.933) .. controls (1.31,-2.49) .. (2.38,-2.125);
        \draw[line width = 1] (1.12,-0.07) .. controls (2.81,0.11) .. (3.03,-1);


        \draw[line width = 1] (2.38,-2.125) --(1.48,-3.684);
        \draw[line width = 1,dotted] (1.48,-3.684) --(1.28,-4.03);
        
        \draw[line width = 1] (3.03,-1) --(3.93,0.559);
        \draw[line width = 1,dotted] (3.93,0.559) -- (4.13,0.9);

        \draw[line width = 1] (1.73,-1) --(2.96,1.12);
        \draw[line width = 1,dotted] (2.96,1.12) --(3.16,1.47);

        \draw[fill] (2.08,-2.645) circle [radius=0.12];
        \draw[fill] (1.78,-3.165) circle [radius=0.12];
        \node at (1.45,-3.165) {$u_5$};

        \draw[fill] (3.33,-0.48) circle [radius=0.12];
        \node at (3.66,-0.48) {$u_4$};
        \draw[fill] (3.63,0.04) circle [radius=0.12];

        \draw[fill] (2.03,-0.48) circle [radius=0.12];
        \draw[fill] (2.63,0.59) circle [radius=0.12];
        
        \draw[fill] (0.433,-0.25) circle [radius=0.12];
        \draw[fill] (1.12,-0.07) circle [radius=0.12];
        \draw[fill] (0.616,-0.933) circle [radius=0.12];
        \draw[fill] (.866,-0.5) circle [radius=0.12];

        \draw[fill] (1.73,-1) circle [radius=0.12];
        \node at (1.45,-1.15) {$u_3$};
        \draw[fill] (2.38,-0.625) circle [radius=0.12];
        \draw[fill] (1.72,-1.75) circle [radius=0.12];
        \draw[fill] (3.03,-1) circle [radius=0.12];
        \draw[fill] (2.38,-2.125) circle [radius=0.12];
        \draw[fill] (3.03,-1.75) circle [radius=0.12];
        \draw[fill] (2.38,-1.375) circle [radius=0.12];

        \draw[line width = 0.5pt] (3.29,-1.9) .. controls (4.694,1.33) .. (3.21,1.55);
        \draw[line width = 0.5pt] (3.21,1.55) .. controls (1.616,0.8) .. (0.2165,-0.125);
        \draw[line width = 0.5pt] (3.29,-1.9) .. controls (1.194,-4.73) .. (0.257,-3.555);
        \draw[line width = 0.5pt] (0.257,-3.555) .. controls (0.116,-1.8) .. (0.2165,-0.125);

        \node at (3.551,-2.05) {\large$B_3$};
        
           \draw[fill] (0,0) circle [radius=0.12];
        \node at (0,-0.25) {$c$};



        \draw[line width = 1] (9,0) -- (9,2);
        \draw[line width = 1] (9,0.5) -- (8.5,1);
        \draw[line width = 1] (9,0.5) -- (9.5,1);

        \draw[line width = 1] (9,2) -- (9.65,2.375) -- (9.65,3.125) -- (9,3.5) -- (8.35,3.125) -- (8.35,2.375) -- (9,2);

        \draw[line width = 1] (8.5,1) .. controls (7.5,2.375) .. (8.4,3.125);
        \draw[line width = 1] (9.5,1) .. controls (10.5,2.375) .. (9.65,3.125);


        \draw[line width = 1] (8.35,3.125) --(6.55,3.125);
        \draw[line width = 1,dotted] (6.55,3.125) --(6.15,3.125);

        \draw[line width = 1] (9.65,3.125) --(11.45,3.125);
        \draw[line width = 1,dotted] (11.45,3.125) -- (11.85,3.125);

        \draw[line width = 1] (9,2) --(11.45,2);
        \draw[line width = 1,dotted] (11.45,2) --(11.85,2);

        \draw[fill] (7.75,3.125) circle [radius=0.12];
        \draw[fill] (7.15,3.125) circle [radius=0.12];

        \draw[fill] (10.25,3.125) circle [radius=0.12];
        \draw[fill] (10.85,3.125) circle [radius=0.12];

        \draw[fill] (9.6,2) circle [radius=0.12];
        \draw[fill] (10.8,2) circle [radius=0.12];
        
        \draw[fill] (9,0.5) circle [radius=0.12];
        \draw[fill] (9.5,1) circle [radius=0.12];
        \draw[fill] (8.5,1) circle [radius=0.12];
        \draw[fill] (9,1) circle [radius=0.12];
        \draw[fill] (9,2) circle [radius=0.12];
        \draw[fill] (9.65,2.375) circle [radius=0.12];
        \draw[fill] (8.35,2.375) circle [radius=0.12];
        \draw[fill] (9.65,3.125) circle [radius=0.12];
        \draw[fill] (8.35,3.125) circle [radius=0.12];
        \draw[fill] (9,3.5) circle [radius=0.12];

        \draw[line width = 0.5pt] (9,3.8) .. controls (12.5,3.4) .. (11.95,2);
        \draw[line width = 0.5pt] (11.95,2) .. controls (10.5,1) .. (9,0.25);
        \draw[line width = 0.5pt] (9,3.8) .. controls (5.5,3.4) .. (6.05,2);
        \draw[line width = 0.5pt] (6.05,2) .. controls (7.5,1) .. (9,0.25);

        \node at (9,4.1) {\large $B_1^*$};


        \draw[line width = 1] (9,0) -- (7.268,-1);
        \draw[line width = 1] (8.567,-0.25) -- (8.384,-0.933);
        \draw[line width = 1] (8.567,-0.25) -- (7.88,-0.07);

        \draw[line width = 1] (7.27,-1) -- (6.62,-0.625) -- (5.97,-1) -- (5.97,-1.75) -- (6.62,-2.125) -- (7.28,-1.75) -- (7.27,-1);

        \draw[line width = 1] (8.384,-0.933) .. controls (7.69,-2.49) .. (6.62,-2.125);
        \draw[line width = 1] (7.88,-0.07) .. controls (6.19,0.11) .. (5.97,-1);


        \draw[line width = 1] (6.62,-2.125) --(7.52,-3.684);
        \draw[line width = 1,dotted] (7.52,-3.684) --(7.72,-4.03);
        
        \draw[line width = 1] (5.97,-1) --(5.07,0.559);
        \draw[line width = 1,dotted] (5.07,0.559) -- (4.87,0.9);

        \draw[line width = 1] (7.27,-1) --(6.04,1.12);
        \draw[line width = 1,dotted] (6.04,1.12) --(5.84,1.47);

        \draw[fill] (6.92,-2.645) circle [radius=0.12];
        \draw[fill] (7.22,-3.165) circle [radius=0.12];

        \draw[fill] (5.67,-0.48) circle [radius=0.12];
        \draw[fill] (5.37,0.04) circle [radius=0.12];

        \draw[fill] (6.97,-0.48) circle [radius=0.12];
        \draw[fill] (6.33,0.59) circle [radius=0.12];
        
        \draw[fill] (8.567,-0.25) circle [radius=0.12];
        \draw[fill] (7.88,-0.07) circle [radius=0.12];
        \draw[fill] (8.384,-0.933) circle [radius=0.12];
        \draw[fill] (8.134,-0.5) circle [radius=0.12];
        \draw[fill] (7.27,-1) circle [radius=0.12];
        \draw[fill] (6.62,-0.625) circle [radius=0.12];
        \draw[fill] (7.28,-1.75) circle [radius=0.12];
        \draw[fill] (5.97,-1) circle [radius=0.12];
        \draw[fill] (6.62,-2.125) circle [radius=0.12];
        \draw[fill] (5.97,-1.75) circle [radius=0.12];

        \draw[line width = 0.5pt] (5.71,-1.9) .. controls (4.306,1.33) .. (5.79,1.55);
        \draw[line width = 0.5pt] (5.79,1.55) .. controls (7.384,0.8) .. (8.7835,-0.125);
        \draw[line width = 0.5pt] (5.71,-1.9) .. controls (7.806,-4.73) .. (8.743,-3.555);
        \draw[line width = 0.5pt] (8.743,-3.555) .. controls (8.884,-1.8) .. (8.7835,-0.125);

        \node at (5.375,-2.05) {\large$B_2^*$};


        \draw[line width = 1] (9,0) -- (10.73,-1);
        \draw[line width = 1] (9.433,-0.25) -- (9.616,-0.933);
        \draw[line width = 1] (9.433,-0.25) -- (10.12,-0.07);

        \draw[line width = 1] (10.73,-1) -- (11.38,-0.625) -- (12.03,-1) -- (12.03,-1.75) -- (11.38,-2.125) -- (10.72,-1.75) -- (10.73,-1);

        \draw[line width = 1] (9.616,-0.933) .. controls (10.31,-2.49) .. (11.38,-2.125);
        \draw[line width = 1] (10.12,-0.07) .. controls (11.81,0.11) .. (12.03,-1);


        \draw[line width = 1] (11.38,-2.125) --(10.48,-3.684);
        \draw[line width = 1,dotted] (10.48,-3.684) --(10.28,-4.03);
        
        \draw[line width = 1] (12.03,-1) --(12.93,0.559);
        \draw[line width = 1,dotted] (12.93,0.559) -- (13.13,0.9);

        \draw[line width = 1] (10.73,-1) --(11.96,1.12);
        \draw[line width = 1,dotted] (11.96,1.12) --(12.16,1.47);

        \draw[fill] (11.08,-2.645) circle [radius=0.12];
        \draw[fill] (10.78,-3.165) circle [radius=0.12];

        \draw[fill] (12.33,-0.48) circle [radius=0.12];
        \draw[fill] (12.63,0.04) circle [radius=0.12];

        \draw[fill] (11.03,-0.48) circle [radius=0.12];
        \draw[fill] (11.63,0.59) circle [radius=0.12];
        
        \draw[fill] (9.433,-0.25) circle [radius=0.12];
        \draw[fill] (10.12,-0.07) circle [radius=0.12];
        \draw[fill] (9.616,-0.933) circle [radius=0.12];
        \draw[fill] (9.866,-0.5) circle [radius=0.12];
        \draw[fill] (10.73,-1) circle [radius=0.12];
        \draw[fill] (11.38,-0.625) circle [radius=0.12];
        \draw[fill] (10.72,-1.75) circle [radius=0.12];
        \draw[fill] (12.03,-1) circle [radius=0.12];
        \draw[fill] (11.38,-2.125) circle [radius=0.12];
        \draw[fill] (12.03,-1.75) circle [radius=0.12];

        \draw[line width = 0.5pt] (12.29,-1.9) .. controls (13.694,1.33) .. (12.21,1.55);
        \draw[line width = 0.5pt] (12.21,1.55) .. controls (10.616,0.8) .. (9.2165,-0.125);
        \draw[line width = 0.5pt] (12.29,-1.9) .. controls (10.194,-4.73) .. (9.257,-3.555);
        \draw[line width = 0.5pt] (9.257,-3.555) .. controls (9.116,-1.8) .. (9.2165,-0.125);

        \node at (12.551,-2.05) {\large$B_3^*$};
        \draw[fill] (9,0) circle [radius=0.12];
        \node at (9,-0.25) {$c^*$};
    \end{tikzpicture}}
    \caption{Graph $G^*$ (left) and induced subgraph $H^*$ (right).}
    \label{fig:IsoEx}
\end{figure}

\begin{theorem}\label{thm:example} For graph $G^*$ and induced subgraph $H^*$  described by Figure~\ref{fig:IsoEx}, $\inoc(G^*,r, 1) < \inoc(H^*, r, 1)$ for every vertex $r$ that appears in both $G^*$ and $H^*$. \end{theorem}

\begin{proof} We first prove that $\inoc(G^*, r, 1)\leq 5$ by considering two cases: whether the root $r$ is $c$ or not.  First, suppose the root is $c$.   If $u_i$ is inoculated at time-step $i$ for $i\in [5]$ as illustrated in Figure~\ref{fig:IsoEx}, the reader can verify 
the contagion will be contained within $5$ time-steps. Second, suppose the root $r$ is not $c$.  Without loss of generality, $r$ is in $B_1$.   If $r$ is a degree-two vertex on an infinite path, the contagion can be contained within two time-steps.  Otherwise, we inoculate $c$ during time-step $1$ to prevent the contagion from spreading to $B_2$ or $B_3$.  For each of the next $3$ time-steps, we inoculate a degree-two vertex on an infinite path that is adjacent to an infected vertex. In doing so, the contagion will be contained within four time-steps.  Thus, $\inoc(G^*, r, 1) \leq 5$.  \medskip

We next prove $\inoc(H^*, r, 1) \geq 6$.  It suffices to show that if the root is $c^*$,  then we will require at least 6 inoculations for containment. Observe that at most one neighbour of $c^*$ can be inoculated during time-step $1$, so the contagion will infect vertices in at least two of $B_1^*,B_2^*$, $B_3^*$. Without loss of generality, a vertex of $B_1^*$ is inoculated during time-step $1$. The contagion then spreads from $c^*$ to its neighbours in $B_2^*$ and $B_3^*$.  Without loss of generality, the vertex inoculated during time-step $2$ is not in $B_3^*$. Then, no vertex of $B_3^*$ is inoculated during time-steps $1$ and $2$.  Containing the contagion on $B_3^*$ requires the inoculation of three additional inoculated vertices, which equates to three additional time-steps.  We must, however, also consider $B_2^*$, in which at most one vertex is inoculated during time-step $2$.  Note that at least one additional vertex must be inoculated in $B_2^*$ in order to contain the virus on $B_2^*$.  Thus, since at least 6 vertices must be inoculated in order to contain the virus on $H^*$, $\inoc(G^*, r, 1) \leq 5 < 6 \leq \inoc(H^*, r, 1)$.\end{proof}   

Since $H^*$ is an induced subgraph of $G^*$, it is not true that $\inoc(H, r, b) \leq \inoc(G, r, b)$ for every induced subgraph $H$ of a graph $G$.  We note that the inequality remains untrue even if the induced subgraph is required to be connected since $H^*$ is a connected induced subgraph of $G^*$.  We also note that the inequality remains untrue even if the subgraph is required to be isometric since $H^*$ is an isometric subgraph of $G^*$. Subgraph $H^*$ is isometric because, for every  pair of vertices $u,v \in V(H^*)$, the distance between $u$ and $v$ in $G^*$ is the same as the distance between $u$ and $v$ in $H^*$.

\section*{Acknowledgements} 
D. Pearson acknowledges research support from the Graph Searching in Atlantic Canada (GSAC) collaborative group, the Atlantic Association for Research in Mathematical Sciences (AARMS), and Mount Allison University. Melissa A. Huggan acknowledges research support from NSERC (grant application 2023-03395 and DGECR-2023-00190). M.E. Messinger acknowledges research support from NSERC (grant application 2018-04059).


\bibliographystyle{siam}

\begin{thebibliography}{10}

\bibitem{anshelevich_2012}
{\sc E.~Anshelevich, D.~Chakrabarty, A.~Hate, and C.~Swamy}, {\em Approximability of the firefighter problem}, Algorithmica, 62 (2012), pp.~520--536.

\bibitem{Av}
{\sc S.~Avadayappan, P.~Jeyanthi, and R.~Vasuki}, {\em Super magic strength of a graph}, Indian Journal of Pure and Applied Mathematics, 32 (2001), pp.~1621--1630.

\bibitem{bazgan_2013}
{\sc C.~Bazgan, M.~Chopin, and B.~Ries}, {\em The firefighter problem with more than one firefighter on trees}, Discrete Applied Mathematics, 161 (2013), pp.~899--908.

\bibitem{Bonato_book}
{\sc A.~Bonato}, {\em An invitation to pursuit-evasion games and graph theory}, American Mathematical Society, Providence, Rhode Island, 2022.

\bibitem{booth_1976}
{\sc K.~S. Booth and G.~S. Lueker}, {\em Testing for the consecutive ones property, interval graphs, and graph planarity using pq-tree algorithms}, Journal of Computer and System Sciences, 13 (1976), pp.~335--379.

\bibitem{cai_2008}
{\sc L.~Cai, E.~Verbin, and L.~Yang}, {\em Firefighting on trees: $(1-1/e)$-approximation, fixed parameter tractability and a subexponential algorithm}, Lecture Notes in Computing Science, 5369 (2008), pp.~258--269.

\bibitem{develin_2004}
{\sc M.~Develin and S.~G. Hartke}, {\em Fire containment in grids of dimension three and higher}, Discrete Applied Mathematics, 155(17) (2007), pp.~2257--2268.

\bibitem{downey_1999}
{\sc R.~G. Downey and M.~R. Fellows}, {\em Parameterized Complexity}, Springer New York, New York, NY, 1999.

\bibitem{finbow_2009}
{\sc S.~Finbow and G.~MacGillivray}, {\em The firefighter problem: A survey of results, directions and questions}, The Australasian Journal of Combinatorics, 43 (2009).

\bibitem{floderus_2011}
{\sc P.~Floderus, A.~Lingas, and M.~Persson}, {\em Towards more efficient infection and fire fighting}, in Proceedings of the Seventeenth Computing on The Australasian Theory Symposium - Volume 119, {CATS} 2011, AUS, 2011, Australian Computer Society, Inc., p.~69–74.

\bibitem{fomin_2016}
{\sc F.~V. Fomin, P.~Heggernes, and E.~J. {van Leeuwen}}, {\em The firefighter problem on graph classes}, Theoretical Computer Science, 613 (2016), pp.~38--50.

\bibitem{G.L.S.2021}
{\sc Z.-B. Gao, G.-C. Lau, and W.-C. Shiu}, {\em Graphs with minimal strength}, Symmetry, 13 (2021).

\bibitem{graham_1994}
{\sc R.~L. Graham, D.~E. Knuth, and O.~Patashnik}, {\em Concrete Mathematics: A Foundation for Computer Science}, Addison-Wesley, Reading, MA, second~ed., 1994.

\bibitem{hartnell_1995}
{\sc B.~L. Hartnell}, {\em Firefighter! an application of domination}, in 25th Manitoba Conference on Combinatorial Mathematics and Computing, University of Manitoba in Winnipeg, Canada, 1995.

\bibitem{Paper2}
{\sc M.~A. Huggan, M.~Messinger, and D.~Pearson}, {\em A note on the strength of a hypercube}, Australasian Journal of Combinatorics, 95(1) (2026), pp.~131--143.

\bibitem{hunter_2024}
{\sc E.~Hunter and J.~Enright}, {\em The {Firefighter} game with state-varying cost functions}, in {Complex Networks \& Their Applications XIII}: Proceedings of The Thirteenth International Conference on Complex Networks and Their Applications: {COMPLEX NETWORKS 2024}, H.~Cherifi, M.~Donduran, L.~M. Rocha, C.~Cherifi, and O.~Varol, eds., vol.~1 of Studies in Computational Intelligence, Istanbul, Turkey, 12 2024, Springer.

\bibitem{I.2018} R.~Ichishima, F.A. Muntaner-Batle, A. Oshima, {\em Bounds for the strength of graphs}, Australasian Journal of Combinatorics, 72 (2018), pp.~492--508.

\bibitem{lampis_2012}
{\sc M.~Lampis}, {\em Algorithmic meta-theorems for restrictions of treewidth}, Algorithmica, 64 (2012), pp.~19--37.

\bibitem{moeller_2002}
{\sc S.~Moeller and P.~Wang}, {\em Fire control on graphs}, JCMCC: The Journal of Combinatorial Mathematics and Combinatorial Computing, 41 (2002), pp.~19--34.

\bibitem{wilson_1996}
{\sc R.~J. Wilson}, {\em Introduction to Graph Theory}, Addison Wesley Longman Ltd., Essex, 4th~ed., 1996.

\end{thebibliography}

\end{document}